\documentclass[11pt]{article}
\usepackage{glauber}
\usepackage{physics}
\usepackage{mathtools}

\usepackage{quasi}
\allowdisplaybreaks

\begin{document}

\title{Locally Stationary Distributions: A Framework for Analyzing Slow-Mixing Markov Chains}
\author{Kuikui Liu\thanks{MIT. Email: \texttt{liukui@mit.edu}}
  \and Sidhanth Mohanty\thanks{MIT. Email: \texttt{sidm@mit.edu}. Supported by CSAIL.}
  \and Prasad Raghavendra\thanks{UC Berkeley. Email: \texttt{raghavendra@berkeley.edu}.
  Supported by NSF CCF-2342192.}
  \and Amit Rajaraman\thanks{MIT. Email: \texttt{amit\_r@mit.edu}. Supported by an Akamai Presidential Fellowship.}
  \and David X. Wu\thanks{UC Berkeley. Email: \texttt{david\_wu@berkeley.edu}. Supported by NSF GRFP DGE-2146752.}}
\date{\today}
\maketitle

\begin{abstract}
Many natural Markov chains fail to mix to their stationary distribution in polynomially many steps. Often, this slow mixing is inevitable since it is computationally intractable to sample from their stationary measure. 

Nevertheless, Markov chains can be shown to always converge quickly to measures that are \emph{locally stationary}, i.e., measures that don't change over a small number of steps.  These locally stationary measures are analogous to local minima in continuous optimization, while stationary measures correspond to global minima.

While locally stationary measures can be statistically far from stationary measures, do they enjoy provable theoretical guarantees that have algorithmic implications?  We study this question in this work and demonstrate three algorithmic applications of locally stationary measures:
\begin{enumerate}
    \item We show that Glauber dynamics on the hardcore model can be used to find independent sets of size $\Omega\left(\frac{\log d}{d} \cdot n\right)$ in triangle-free graphs of degree at most $d$. 

    \item Let $W$ be a symmetric real matrix with bounded spectral diameter and $v$ be a unit vector. Given the matrix $M = \lambda vv^\top + W$ with a planted rank-one spike along vector $v$, for sufficiently large constant $\lambda$, Glauber dynamics on the Ising model defined by $M$ samples vectors $x \in \{\pm 1\}^n$ that have constant correlation with the vector $v$.  
    \item Let $M = A_{\bG} - \frac{d}{n}\bone\bone^\top$ be a centered version of the adjacency matrix where the graph $\bG$ is drawn from a sparse 2-community stochastic block model with signal-to-noise ratio $\lambda$.
    We show that for sufficiently large constant $\lambda$, Glauber dynamics on the Ising model defined by $M$ samples vectors $x \in \{\pm 1\}^n$ that have constant correlation with the hidden community vector $\bsigma$.
\end{enumerate}
In other words, Glauber dynamics subsumes the spectral method for \emph{spiked Wigner} and \emph{community detection}, by \emph{weakly} recovering the planted spike.
\end{abstract}

\thispagestyle{empty}
\setcounter{page}{0}
\newpage
\tableofcontents
\thispagestyle{empty}
\setcounter{page}{0}
\newpage


\section{Introduction}

Markov chains are a fundamental algorithmic primitive that are widely applied towards sampling and counting tasks.
There is a rich body of literature devoted to understanding worst-case \emph{mixing times} of Markov chains, i.e., the number of steps required for the distribution of the chain to approach its stationary measure started from an arbitrary initialization. For some highlights in this area, see, e.g., the contents and references in \cite{MT06,Dia09,BGJM11,BGL14,LPW17}.

Unfortunately, many natural Markov chains fail to mix rapidly from worst-case initializations, in that it takes a super-polynomial number of steps to reach stationarity.
Structurally, this is due to the presence of cuts in the state space with very small conductance.
Often, slow mixing is inevitable since sampling from the stationary measure is known to be computationally hard (say $\NP$-hard); see, e.g., \cite{Sly10,SS14,GV16}.
Nevertheless, it has been empirically observed that certain simple and local Markov chains like \emph{Glauber dynamics} succeed at optimization and inference tasks even when they are not known to mix, such as finding satisfying assignments to SAT formulas \cite{SKC94, BILSZ16}, and clustering stochastic block models \cite{MS12, GBP19}.
This suggests that local Markov chains like Glauber dynamics can have algorithmic applications, even if they \emph{fail} to mix rapidly, which raises our main line of inquiry.
\begin{question}
    What is the long-term behavior of Markov chains that do not mix rapidly?
    How can slow-mixing Markov chains be harnessed for optimization and inference?
\end{question}
To lay out the motivation, it is useful to draw an analogy to continuous optimization. Gradient descent, the canonical algorithm in optimization, converges efficiently to a global minimum if the objective function and parameter space is convex. However, non-convex objective functions and parameter spaces come up often both in theory and practice, and finding global minima can even be provably intractable. 
On the other hand, gradient descent can be shown to always converge quickly to a local minimum, or more precisely, a first-order stationary point. Moreover, these local minima are useful in practice, and can admit non-trivial theoretical guarantees.
See \cite{Nes18} for a comprehensive coverage of analyzing gradient-based optimization methods, and see, e.g., \cite{GJZ17,BLGT20,JT20,JNGMJ21} and the references therein for non-trivial theoretical guarantees on local minima of gradient descent in the context of machine learning.

Analogously, in the context of sampling, certain random walks can be shown to mix rapidly to their stationary measures. 
However, other random walks can also be shown to mix slowly from worst-case initializations; these are akin to the ``hard'' nonconvex optimization problems.
Despite this, one can show that \emph{any} Markov chain satisfying fairly generic conditions converges to analogs of local minima that we term \emph{locally stationary distributions} (see \pref{def:localstationary}).  
Intuitively, a locally stationary distribution corresponds to the stationary measure conditioned on a subset of states that are sparsely connected to the rest of the state space under the Markov chain.

This raises the question of whether these locally stationary measures obey theoretical guarantees that are useful for solving problems in optimization or inference.  

\subsection{Locally stationary distributions}
Let $P$ be the transition matrix of a time-reversible Markov chain on a state space $\Omega$, where $P[x,y]$ denotes the transition probability from $x$ to $y$.  Let $\pi$ be a stationary distribution w.r.t. $P$.
 %

Analogous to local minima in optimization, a locally stationary measure $\nu$ is one started at which the Markov chain $P$ remains nearly stationary, i.e., makes little progress.
We will use the \emph{KL divergence} to the stationary distribution, denoted $\KL{\nu}{\pi}$, as a measure of progress.  
This leads to the following definition of an $\eps$-locally stationary measure.

\begin{definition} \label{def:localstationary}
A probability measure $\nu$ on $\Omega$ with density $\density$ relative to $\pi$ is said to be \emph{$\eps$-locally stationary} with respect to $P$ if 
$$
    \calE(\density,\log\density) \coloneqq \sum_{x,y \in \Omega} P[x,y] \cdot \parens*{\density(x) - \density(y)} \cdot \log\frac{\density(x)}{\density(y)} \leq \eps
$$
\end{definition}

The \emph{Dirichlet form} $\calE(\density,\log \density)$ measures the rate at which the Markov chain progresses towards the stationary distribution.  In particular, for a continuous-time version of the Markov chain $P$, if $\nu_t$ denotes the measure at time $t$, we have the following well-known fact \cite{BT06}:
$$\frac{\dif}{\dif t}\KL{\nu_t}{\pi} = -\calE(\density_t,\log\density_t).$$ 

As an immediate consequence, we observe that the Markov chain is typically on $\eps$-locally stationary measures over time.  Formally, we have the following claim.

\begin{restatable}{theorem}{mainls}
    \label{thm:main-ls}
    Fix a time-reversible Markov chain $P$ with a stationary measure $\pi$, any starting distribution $\nu_0$, and $\eps,\delta > 0$.
    Let $T = \frac{1}{\delta\eps} \cdot \log\left( \frac{1}{\pi_{\min}} \right)$.
    Then, for a time $\bt \sim [0,T]$ chosen uniformly at random, the distribution $\nu_{\bt}$ at time $t$ is $\eps$-locally stationary with respect to $P$ with probability at least $1-\delta$.
\end{restatable}
The main conceptual contribution of our work is the following meta-principle for showing that locally stationary distributions solve optimization and inference problems.
\begin{displayquote}
    \emph{Prove that sampling from the true stationary distribution solves the optimization or inference problem of interest, and additionally, does so for ``local'' reasons.}
\end{displayquote}
This principle is best illustrated by discussing our algorithmic applications of locally stationary distributions.

\parhead{Independent sets in triangle-free graphs.}
It is easy to see that any graph $G$ on $n$ vertices with degree bounded by $d$ has an independent set of size $\frac{n}{d+1}$, a bound which is tight for the union of $(d+1)$-sized cliques.
Ajtai, Koml{\'o}s, and Szemer{\'e}di \cite{AKS80} showed that when $G$ is triangle-free, the size of the maximum independent set guaranteed to exist increases to $\Omega\parens*{n\cdot\frac{\log d}{d}}$.
Shearer \cite{She83} gave an alternate proof which pins down the leading constant to $1 - o_d(1)$ and relaxes the assumption of bounded maximum degree to bounded average degree.

It is also well-known that a uniformly random independent set in a triangle-free graph of maximum degree $d$ has expected size at least $\Omega\parens*{n\cdot\frac{\log d}{d}}$; see, e.g., \cite[Proposition 1, Page 272]{AS16}. Hence, it is natural to wonder whether \emph{Glauber dynamics} with respect to the uniform measure over independent sets finds such a large independent set. From a given independent set $I \subseteq V$, the transitions of Glauber dynamics can be described as follows:
\begin{enumerate}
    \item Sample a uniformly random vertex $v \in V$.
    \item If $I \cup \{v\}$ is an independent set, then go to $I' = I \cup \{v\}$ with probability $1/2$ and $I' = I \setminus \{v\}$ with probability $1/2$.
    \item If $I \cup \{v\}$ is not an independent set, go to $I' = I \setminus \{v\}$ with probability $1$.
\end{enumerate}

Notably, this Markov chain requires $\exp(\Omega(n))$ steps to mix \cite{MWW07} as soon as $d \geq 6$. In fact, the problem of sampling a uniformly random independent set on a graph of maximum degree $d$ becomes $\NP$-hard \cite{Sly10, SS14} in this regime, even if triangle-freeness is assumed \cite{GSV15}.

Despite these hardness results for the corresponding sampling problems, we show that the above Markov chain can be used to find independent sets of size $\Omega\parens*{n\cdot\frac{\log d}{d}}$ in triangle-free graphs of maximum degree bounded by $d$.
Specifically, we show the following result.
\begin{restatable}{theorem}{glauberindepset}\label{thm:glauber-indepset}
    Let $G$ be a triangle-free graph on $n$ vertices with maximum degree bounded by $d$.
    Let $\bI$ be an independent set in $G$ that arises from Glauber dynamics run for $O\parens*{nd^4}$ time.
    Then the expected size of $\bI$ is at least $\frac{1-o_d(1)}{4}\cdot n\cdot\frac{\log d}{d}$.
\end{restatable}
\begin{remark}
    In fact, one can prove that Glauber dynamics at ``fugacity'' $\tfrac{1}{\log d}$ finds an independent set of size $(1-o_d(1))\cdot n\cdot\frac{\log d}{d}$ by combining our proof method with that of \cite{DJPR18}.
\end{remark}

As mentioned before, we know that the expected size of a uniformly random independent set satisfies the above lower bound.
However, the Glauber dynamics chain does not mix rapidly, and hence does not produce samples from the truly uniform distribution.
Instead, it samples from a locally stationary distribution with respect to the Markov chain.
Our key insight is that the same proof also goes through for an independent set sampled from a locally stationary distribution with respect to Glauber dynamics.

To give a sense of how local stationarity is used, we briefly discuss the proof.
The proof from \cite{AS16} that the expected size of a uniformly random independent set $\Omega\parens*{n\cdot\frac{\log d}{d}}$ argues that for any vertex $v$, and for any pinning $x_{v\cup N(v)}$ of the independent set outside $v$ and its neighbors, either:
\begin{itemize}
    \item the uniform distribution conditioned on the pinning chooses $v$ with probability $\gtrsim \frac{\log d}{d}$, or
    \item it chooses $\gtrsim \log d$ neighbors of $v$ in expectation.
\end{itemize}
Thus, each vertex can be charged $\Omega\left(\frac{\log d}{d}\right)$ vertices on average in the independent set.
Observe that the above sketch of the argument goes through even if the distribution is not truly uniform but merely has conditional marginals matching the uniform distribution, which is a property we can show holds for locally stationary distributions (\pref{lem:localpatches}).

Using similar arguments, one can establish that given a triangle-free graph with maximum degree $d$, Glauber dynamics run for $\poly(n)$ many steps on the antiferromagnetic Ising model on $G$ with inverse temperature $\frac{1}{\sqrt{d}}$ recovers a cut of relative size $\frac{1}{2}+ \Omega\left(\frac{1}{\sqrt{d}}\right)$.

\medskip

\parhead{Weak recovery in spiked models.} Beyond independent sets, we also study the performance of Glauber dynamics for statistical inference tasks. Consider the central class of Bayesian models for principal component analysis (PCA) known as \emph{spiked random matrix models}, which consist of a matrix $M \in \R^{n \times n}$ given by
\begin{align*}
    M = \underset{\substack{\uparrow \\ \text{signal} \\ \text{strength}}}{\lambda} \cdot \underset{\substack{\uparrow \\ \text{signal}}}{vv^\top} + \underset{\substack{\uparrow \\ \text{noise}}}{W}.
\end{align*}
The general algorithmic question is to approximately recover the signal (a unit vector $v$) under appropriate assumptions about the noise ($W$) and the signal strength ($\lambda$).
More precisely:
\begin{problem}[Weak recovery in spiked matrix models]
    For a unit norm signal vector $v \in \R^n$, signal strength $\lambda \in \R$, and noise matrix $W \in \R^{n \times n}$, given $M = \lambda \cdot vv^{\top} + W$, give an efficient algorithm to extract a unit norm estimate $\wh{v}$ such that $\angles{v,\wh{v}} \geq \Omega(1)$.
\end{problem}

In the situation where $W$ is a Wigner matrix, this model, known as the \emph{spiked Wigner model}, has been a subject of extensive study.
The work of \cite{BBP05} determined that once $\lambda > 1$, a spectral algorithm based on computing the top eigenvector succeeds at weak recovery.  There is a fairly large body of work on the spiked Wigner 
model, towards characterizing optimal estimation error, efficient algorithms and its generalizations to rank larger than one \cite{DMK16,DMKLZ16,DAM17,EAK18,LM17,MIO17,BM19,MV21}.
When the prior distribution over $v$ is the uniform distribution over $\left\{\pm \tfrac{1}{\sqrt{n}}\right\}^n$ and $\bbS^{n-1}$, there are efficient algorithms that even achieve the maximum information-theoretically achievable correlation $\abs*{\angles*{v, \wh{v}}}$, based on approximate message passing \cite{FMM18,CFM23}, and algorithmic stochastic localization \cite{MW23}.


In the case where the prior distribution is on the hypercube, given the matrix $M$, this posterior is described by an Ising model: a probability distribution $\mu_{\beta M}$ over $\{\pm 1\}^n$ defined by the following proportionality relation for a suitably chosen $\beta > 0$:
$$ \mu_{\beta M}(x) \propto \exp(\tfrac{1}{2}\angles{x, \beta M x}) \text{ for all } x \in \{\pm 1\}^n. $$

Sampling from the above posterior distribution is desirable as it achieves the maximum information-theoretically achievable correlation.
The canonical algorithm for sampling is to run the Glauber dynamics Markov chain, but unfortunately, provable guarantees for Glauber dynamics are currently lacking.
Thus, a natural question en route is: \emph{does Glauber dynamics for the Ising model $\mu_{\beta M}$ weakly recover the signal in polynomial-time?}

We make progress towards answering this question affirmatively in this work by showing that Glauber dynamics at a slightly higher temperature than the posterior distribution succeeds for a broad family of settings.

Formally, we show the following result:
\begin{restatable}{theorem}{mainspikedwig}
    \label{th:correlation-gain}
    Let $W$ be a matrix with $\kappa \psdle W \psdle 1-\kappa$, and $v \in \left\{ \pm \frac{1}{\sqrt{n}}\right\}^n$. Let $P$ denote the kernel of the Glauber dynamics chain with stationary distribution $\mu_{W + \lambda vv^\top}$, and $x_0$ an arbitrary point on $\{\pm 1\}^n$.
    There exists a large enough constant $\lambda > 0$ such that for $T = \wt{\Theta}(n^{5})$, and for $\bt\sim[0,T]$, with probability $1-o(1)$ we have:
    \[
        \E_{\bx \sim P^t\delta_{x_0}}[\left|\langle \bx,v \rangle\right|] \ge \left(\kappa \exp\parens*{ - \frac{1}{\kappa}} - o(1) \right) \cdot \sqrt{n}\mper
    \]
\end{restatable}

A natural approach to recover the signal $v$ from the matrix $M$ is the spectral method, which amounts to computing (even approximately) the eigenvector corresponding to the largest eigenvalue for the matrix $M$.
At a high-level, the above theorem demonstrates that Glauber dynamics can \emph{simulate} the spectral method in certain regimes. 
We further expect Glauber to achieve weak recovery when run for $T = n^{1 + o(1)}$ steps, but we leave this open as a direction for future improvement.

\begin{remark}
    The above model of choice captures several commonly considered models of study in the algorithms and complexity of statistical inference, such as the \emph{spiked Wigner model} \cite{BBP05}, and \emph{random/planted 2XOR} (see, e.g., \cite{AOW15} and the references within).
\end{remark}

\begin{remark}\label{rem:spiked-wig-posterior}
    In the Rademacher spiked Wigner model, where $W \sim \GOE(n)$ and $v \sim \left\{\pm \frac{1}{\sqrt{n}} \right\}^n$, the posterior has the form
    \[ \Pr[v|M] \propto \exp\left( - \frac{n}{2} \left\| M - \lambda vv^\top \right\|_F^2 \right) \propto \exp\left( \frac{\lambda n}{2} \cdot v^\top M v \right). \]
    The above is an Ising model, and suggests Glauber dynamics as a natural algorithm for weak recover. We note that the Ising model we shall analyze will be a higher temperature version of the above, that is, a distribution with density proportional to $\exp\left( \frac{\beta n}{2} \cdot v^\top M v \right)$ for some $\beta < \lambda$ (as opposed to the ``correct'' value $\lambda$). Interestingly, such recovery guarantees were not previously known, even if one \emph{information theoretically} samples from the higher temperature Ising model.
\end{remark}

\parhead{Stochastic block model.}
Another case of interest is one where the Ising model $M$ arises from a stochastic block model.  
To describe this result, we first begin by defining the two-community stochastic block model.
\begin{definition}[$2$-community stochastic block model]
    Let $d, \lambda \in \R$ be fixed parameters such that $\lambda^2 \le d$. 
    The distribution $\SBM(n, d, \lambda)$ is defined over pairs $(\sigma, G) \in \{\pm1\}^n \times \{0, 1\}^{n \times n}$ generated as follows.
    
    Let $\bsigma \in \{\pm1\}^{n}$ be a \emph{signal} vector drawn uniformly at random (i.e. the prior is uniform). Given $\bsigma$, we draw a random graph $\bG$ by including an edge between $u,v \in [n]$ independently with probability $\frac{d + \lambda \sqrt{d}}{n}$ if $\bsigma(u) = \bsigma(v)$, and with probability $\frac{d - \lambda\sqrt{d}}{n}$ otherwise.
\end{definition}
In a general stochastic block model, the signal vector $\bsigma$ can be over a larger finite alphabet $[q]$, and the probability of including an edge between $u,v \in [n]$ is an arbitrary function of $\sigma(u), \sigma(v)$.
In this work, we will use the term stochastic block model (SBM) to refer exclusively to the special case of two communities as defined above.

\begin{remark}
    The $2$-community stochastic block model can be viewed as a special case of a spiked matrix model where $M$ is a highly sparse matrix.
    Due to the sparsity of $M$, this spiked matrix model falls outside the scope of \Cref{th:correlation-gain}, as the ``noise'' part fails to satisfy the spectral bound. 
\end{remark}

The weak recovery problem for stochastic block model is that of recovering a labelling $\wh{\sigma}$ given the graph $\bG$ such that $\wh{\sigma}$ has non-trivial correlation with the true signal $\bsigma$.  More precisely, an algorithm for weak recovery is required to find a $\wh{\sigma}$ such that $\frac{1}{n} |\angles{\wh{\sigma}, \bsigma}| \geq \Omega(1)$.

Starting with the work of Decelle, Krzakala, Moore, \& Zdeborova \cite{DKMZ11} that posited broad conjectures about these models, an extensive body of work has emerged over the past decade.  For the case of $2$ communities, \cite{DKMZ11} posited that weak-recovery is possible if and only if the signal strength $\lambda^2 > 1$.  This coincides with the  \emph{Kesten--Stigum threshold}, a threshold for broadcast processes on trees studied in the works of Kesten and Stigum \cite{KS66, KS67}.
The works of Mossel, Neeman, \& Sly \cite{MNS18} and Massouli{\'e} \cite{Mas14} confirmed the algorithmic side, namely that weak recovery can be solved efficiently above the KS threshold with a spectral algorithm, while \cite{MNS15} showed impossibility below the threshold.
We refer the reader to the survey of Abbe \cite{Abb17} for a detailed treatment of the literature on community detection.

We show that Glauber dynamics succeeds at weak recovery when the signal strength is a constant factor above the Kesten--Stigum threshold.
\begin{restatable}{theorem}{mainsbm}
    \label{th:sbm-recovery}
    There exist constants $\lambda_0, \beta, c > 0$ such that for all $\lambda$ satisfying $\abs{\lambda} \ge \lambda_0$, for $(\bsigma, \bG) \sim \SBM(n, d, \lambda)$, with probability $1-o(1)$ over the randomness of $(\bsigma, \bG)$, the following holds.

    Let $P$ denote the kernel of Glauber dynamics with stationary distribution $\mu_{\frac{\beta}{\sqrt{d}}\left(A_{\bG} - \frac{d}{n}\bone\bone^\top\right)}$ and $x_0$ an arbitrary point on $\{\pm 1\}^n$.
    For $T = \wt{\Theta}\parens*{n^{5+o_d(1)}}$, and for $\bt\sim[0,T]$, with probability $1-o(1)$, we have:
    \[
        \E_{\bx\sim P^{\bt} \delta_{x_0}} \bracks*{ |\angles*{\bx,\bsigma}| } > c n.
    \]
\end{restatable}
\begin{remark}\label{rem:sbm-posterior}
    Similar to \Cref{rem:spiked-wig-posterior}, the posterior distribution $\bsigma|\bG$ to solve the recovery problem in the stochastic block model is an Ising model (see, e.g., \cite[Eq. (7)]{Moo17}):
    \begin{align*}
        \Pr[\bsigma|\bG] &\propto \prod_{ij\in E(\bG)} \parens*{\frac{d+\lambda\sqrt{d}}{d-\lambda\sqrt{d}}}^{(1+\bsigma_i\bsigma_j)/2} \cdot \prod_{ij\notin E(\bG)} \parens*{ \frac{1-\frac{d+\lambda\sqrt{d}}{n}}{1-\frac{d-\lambda\sqrt{d}}{n}} }^{(1+\bsigma_i\bsigma_j)/2} 
    \end{align*}
    For large $d$, the Ising model that we analyze is approximately equal to a higher temperature version of the above true posterior.
\end{remark}

Although spectral algorithms for weak recovery were already known in all the cases listed above, understanding the power of Glauber dynamics is interesting in its own right.  
It is arguably a more natural algorithm than spectral methods in the context of a Bayesian estimation problem like stochastic block models.  
In particular, Glauber dynamics remains locally consistent with the underlying probabilistic model at every vertex.
On the other hand, a spectral algorithm that computes the top eigenvector maximizes a global objective, while crudely approximating the local features of the probabilistic model.

Finally, our analysis for spiked models establishes a direct correspondence between locally stationary measures for Glauber dynamics and fixed points of a Markov chain over the one-dimensional real line $\R$ (related to the {\it restricted Gaussian dynamics} Markov chain).  
This correspondence may pave the way for a much tighter analysis to establish that Glauber dynamics achieves information theoretically optimal recovery in some of these models.
To elucidate further on this correspondence, we will give a brief technical overview here.

\subsection{Technical overview}
In \prettyref{sec:lsm-props}, we derive a few basic properties of locally-stationary measures.  This is followed by the result on independent sets presented in  \cref{sec:glauber-indepset} as a warmup.

In this technical overview, we will focus on the inference problem in spiked matrix models.
For a matrix $M \in \R^{n \times n}$ and a vector $h \in \R^{n}$, we will use $\mu_{M,h}$ to denote the distribution over $\{\pm 1\}^n$ defined as
$$ \mu_{M,h}(x) \propto \exp(\tfrac{1}{2}\angles{x, M x} +  \angles{h, x}),$$
and $\mu_{M}$ to denote the distribution $\mu_{M, 0}$.

Consider the stationary measure $\mu_{M}$ for a spiked matrix $M = \lambda vv^\top + W$.
We outline the proofs of \Cref{th:correlation-gain,th:sbm-recovery} here, which consist of two parts.
\begin{itemize}
    \item First, we show that locally stationary distributions with respect to Glauber dynamics over $\{\pm 1\}^n$ are also locally stationary with respect to the \emph{restricted Gaussian dynamics} Markov chain.
    \item Next, we show that samples from locally stationary distributions for RGD achieve weak recovery.
\end{itemize}
Let us first recall the definitions of Glauber dynamics and Restricted Gaussian dynamics for Ising models. 
\begin{definition}[Glauber dynamics]
    Glauber dynamics with respect to a distribution $\pi$ over $\{\pm 1\}^n$ is a Markov chain on $\{\pm 1\}^n$, where a transition from $x$ is given by the following:
    \begin{itemize}
        \item Sample index $i$ uniformly from $[n]$.
        \item Transition to $x^{\oplus i}$ with probability $\frac{\pi(x^{\oplus i})}{\pi(x) + \pi(x^{\oplus i})}$, and stay at $x$ otherwise. Here, $x^{\oplus i}$ denotes $x$ with the $i$th bit flipped.
    \end{itemize}
\end{definition}
\begin{definition}[{Restricted Gaussian dynamics; cf.\ \cite{LST21,STL20,CE22}}]
    Consider the joint random variable $(\bx,\bz)$ where $\bx\sim\mu_M$, and $\bz|\bx \coloneqq (\lambda \angles*{v,\bx} + \sqrt{\lambda}\bg) \cdot v$ for $\bg\sim\calN(0,1)$.
     \emph{Restricted Gaussian dynamics} (RGD) is a Markov chain on $\{\pm1\}^n$ where for any $x$, the transition to $\bx'$ is sampled as follows:
    \begin{itemize}
        \item Sample $\bz|x$. 
        \item Sample $\bx'|\bz$.
    \end{itemize}
\end{definition}
\begin{remark}
    By definition, the above Markov chains are ergodic and reversible with respect to $\pi$ and $\mu_M$ respectively, and so  asymptotically converge to them as their stationary distributions.
\end{remark}
\begin{remark}
    We should think of $\bz|\bx$ as being a noisy surrogate for how well $\bx$ correlates with the hidden direction $v$.
\end{remark}
Informally, we prove the following correspondence between locally stationary distributions for Glauber dynamics and locally stationary distributions for RGD.
In fact, this correspondence is a consequence of a more generic statement; refer to \Cref{lem:ls-dist-rgd} for details.
\begin{lemma}[{Informal version of \Cref{lem:ls-dist-rgd}}]
    Let $\nu$ be a distribution over $\{\pm1\}^n$ that is $\eps$-locally stationary under Glauber dynamics for $\mu_M$.
    Suppose for every $z\in\R$, Glauber dynamics for the distribution of $\bx|z$ is ``well-expanding'', and $\log\frac{1}{\mu_M(x)}\le\poly(n)$ for all $x\in\{\pm1\}^n$, then
    $\nu$ is $\eps\cdot\poly(n)$-locally stationary under restricted Gaussian dynamics.
\end{lemma}
To conclude that $\nu$ is locally stationary under restricted Gaussian dynamics, it suffices to verify the structural properties of $\mu_M$.
The lower bound on the minimum probability follows from upper and lower bounds on the values that the Hamiltonian can achieve.
To show that Glauber dynamics for $\bx|\bz$ is well-expanding, we must investigate the structure of this distribution further.
A simple calculation reveals that the distribution of $\bx|\bz$ is, in fact, the Ising model $\mu_{W,\bz}$.
In the setting of \Cref{th:correlation-gain}, where the spectral diameter of $W$ is bounded by $1$, prior works \cite{BB19,EKZ22,AJKPV21,CE22} prove that $\mu_{W,\bz}$ always satisfies a ``modified log-Sobolev inequality'' (our relevant notion of ``well-expanding'').
In the setting of \Cref{th:sbm-recovery}, where $W$ is a centered stochastic block model, a similar result is proved in a companion paper \cite{LMRW}.
\begin{remark}
    This decomposition of $\mu_M$ into a mixture of other Ising models is well-known in the literature by the name \emph{Hubbard--Stratonovich transformation} \cite{HS59}; see also \cite{KLR22}.
\end{remark}

In summary, we showed if $\nu$ is locally stationary with respect to Glauber dynamics, then it is also locally stationary for the RGD chain
\[
    \bx \to \bz|\bx \to \bx'|\bz.
\]

Thus, it suffices to prove that $\E_{\bx\sim\nu}|\angles*{\bx,v}|$ is bounded away from $0$ for any distribution $\nu$ that is locally stationary for RGD.
The two ingredients that go into proving this are:
\begin{itemize}
    \item A generic principle that says: if $\nu$ is locally stationary for a Markov chain $P$, then for any bounded function $f$, $\left|\E_{\bx\sim\nu} f(\bx)-\E_{\bx\sim P\nu} f(\bx)\right|$ is small (\Cref{cor:bounded-function-stability}).
    \item If the correlation of a distribution $\nu$ is too close to $0$, then a single step of RGD causes a significant boost in correlation, which means $\nu$ cannot be locally stationary.
    In particular, for $f(x) = |\angles*{x, v}|$, if $\E_{\bx\sim\nu} f(\bx)$ is too close to $0$, then $\E_{\bx\sim P\nu} f(\bx) - \E_{\bx\sim\nu} f(\bx)$ is nontrivially large, which means any locally stationary distribution $\nu$ must achieve large correlation.
\end{itemize}
See \pref{sec:weak-recovery} for the details of this argument.

\subsection{Related work}
Motivated by statistical physics,
the phenomenon of \emph{metastability} of random walks has been extensively studied.
We refer the reader to the monograph by Bovier \& Hollander \cite{BOV2016} for related literature. 
The notion of metastability in \cite{BOV2016} appears to be a slightly stricter notion than local stability, and thus does not generically hold for all reversible Markov chains.

In the context of sampling distributions over a continuous domain, Balasubramanian, Chewi, Erdogdu, Salim \& Zhang \cite{BCESZ22} showed that the Langevin Monte Carlo algorithm always outputs a sample from a distribution whose relative Fisher information is small.  This is the continuous sampling analog of convergence of gradient descent to approximate first-order stationary points. 
Building on these ideas, Cheng, Wang, Zhang \& Zhu \cite{CWZZ24} study the notion of \emph{conditional mixing} for Langevin and Glauber dynamics and apply it to efficiently sample from Gaussian mixtures.

Our analysis of Glauber dynamics borrows ideas from a recent line of works on sampling from Ising models.
Glauber dynamics for an Ising model defined by a matrix $M$ was shown to mix quickly if eigenvalues of $M$ lie within an interval of length $1$ \cite{AJKPV21,EKZ22}. This is sharp, as evidenced by the Curie--Weiss model $M = \frac{\beta}{n} \bone\bone^{\top}$. Stronger evidence for hardness of sampling beyond this spectral criterion was recently provided by Kunisky \cite{Kun23}, based on a reduction to a certain statistical hypothesis testing problem.
Koehler, Lee \& Risteski \cite{KLR22} devised more sophisticated algorithms based on simulated tempering and variational inference to sample from Ising models when they have constantly many eigenvalues outside an interval of length $1$.

Besides the question of fast mixing and metastability, the problem of how well MCMC-based algorithms perform for optimization and inference tasks was recently studied in several works. Chen, Mossel \& Zadik \cite{CMZ23} proved that when initialized at the empty set, natural Metropolis chains on cliques fail to find cliques of sublinear size in polynomial time, even if such a clique is planted inside the Erd\H{o}s--R\'{e}nyi random graph $\ER(n,1/2)$. This is despite there being an abundance of algorithms which can recover a planted clique of size down to $O\parens*{\sqrt{n}}$.
Nevertheless, MCMC-based algorithms were redeemed in a more recent work of Gheissari, Jagannath \& Xu \cite{GJX23} using a more carefully designed chain and initialization.
In a recent work, Sellke \cite{Sel23} proved that low-temperature Langevin dynamics achieves the conjectured computational threshold for optimizing pure spherical spin glass models.

\subsection{Open problems}
We conclude with several open directions, which we believe may be amenable to the framework of locally stationary distributions.

\parhead{Bayesian inference via MCMC.}
First, there is the direction of pushing our results further in the setting of SBM.
To set the scene, let $\pi(x) \propto \exp(H(x))$ be the true posterior for SBM, where $H(x)$ is the SBM Hamiltonian (see \Cref{rem:sbm-posterior} for an explicit formula).
It is well known that optimal recovery is achieved information theoretically by sampling from $\pi$ (see, e.g., \cite[Section 4]{Moo17}).
However, it takes exponential time to mix to $\pi$ from a worst-case initialization.
On the other hand, in \Cref{th:sbm-recovery}, we achieve weak recovery by running Glauber dynamics on the density $\pi_{\beta}(x) \propto \exp(\beta H(x))$ for some (constant) $\beta$ strictly smaller than 1.
It is natural to investigate whether a sampling algorithm based on \emph{simulated annealing}, i.e. running Glauber dynamics by varying the temperature over time, can succeed at sampling from $\pi$.
Our main result can be viewed as a modest step in this direction, as we show that running the chain for $\poly(n)$ steps at a mismatched temperature already gives a warm start for correlation.
\begin{problem}[Optimal recovery for stochastic block model]
    Can an instance of simulated annealing sample from $\pi$?
\end{problem}

\parhead{Computationally optimal inference.}
The $k$-community stochastic block model is known to undergo an \emph{information-computation gap} when $k \ge 5$ (see, e.g., \cite{AS16}).
Specifically, for every $k\ge 5$, there exists a choice of degree $d$ and signal-to-noise ratio $\lambda$ for which weak recovery is information-theoretically possible, but likely impossible for efficient algorithms \cite{HS17}.
This gap admits the construction of SBM instances where weak recovery is tractable, but information-theoretically optimal recovery is intractable to efficient algorithms.
\begin{example}
    Consider a $10$-community block model obtained by taking two disjoint $5$-community block model graphs, and planting a sparse bisection between them.
    The planted bisection should be sparse enough so it is clearly detectable to efficient algorithms.
    However, the parameters $d$ and $\lambda$ for the $5$-community models are chosen to be in the intractable regime.
    An information-theoretically optimal algorithm achieves weak recovery within each $5$-community model.
    At the same time, algorithms for the $2$-community block model achieve weak recovery in the $10$-community model, since they can find the planted bisection, and correctly classify vertices as belonging to either $\{1,2,3,4,5\}$ or $\{6,7,8,9,10\}$.
\end{example}
In settings such as the above, information-theoretically optimal inference is hard, but weak recovery is still tractable.
This motivates the study of \emph{computationally optimal inference algorithms}: algorithms that achieve the best guarantees possible in polynomial time.
For such problems, the Glauber dynamics chain must necessarily fail to mix rapidly to the posterior distribution, but perhaps the locally stationary distribution it samples from can achieve the computationally optimal recovery guarantees?
\begin{problem}[Computationally optimal inference]
    Is (annealed) Glauber dynamics a computationally optimal algorithm for the $k$-community SBM and, more generally, for random CSPs with planted solutions?
\end{problem}

\parhead{Metastable states.}
Local stationarity is a generic property of any time-reversible Markov chain, so a priori there is no reason to expect that a locally stationary distribution $\nu$ has any nice properties.
For example, if we run the Markov chain for $T$ steps, \Cref{thm:main-ls} guarantees an $\eps$-locally stationary distribution $\nu$ where $\eps = O(1/T)$, and the simple random walk on the $n$-vertex cycle graph demonstrates that this is tight if $T = o(n^2)$.
This suggests the following natural questions. Suppose the stationary distribution $\pi$ is a  a Gibbs distribution on $\{\pm 1\}^n$. 
Under what additional structural assumptions on $\pi$ can we both obtain $\eps = o(1/T)$ for sufficiently large $T = \poly(n)$ and endow $\nu$ with a physical or geometric interpretation?  

The notion of metastable states for Gibbs distributions in statistical physics \cite{BOV2016} suggests a conceptual path forward towards this goal. 
In particular, one might hope to show that a locally stationary distribution is close to a metastable state, i.e. a conditional Gibbs distribution restricted to a metastable set of configurations. 
\begin{problem}[Metastability]
    Suppose $\pi$ is a Gibbs distribution which has a metastable subset $S$ with exponentially small conductance.
    Let $\nu$ be the locally-stationary distribution after running Glauber (or Langevin) dynamics for $\poly(n)$ steps with uniform initialization in $S$. 
    Is $\nu$ close to the conditional Gibbs distribution $\pi_S$, e.g., $\KL{\nu}{\pi_S} = o(1)$? 
\end{problem}
For a concrete setting, suppose $\pi$ is a spherical spin glass in the \emph{shattering} regime \cite{EAMS23}. Can one show that in $\poly(n)$ time, Langevin dynamics with uniform initialization remains stuck in the clusters identified there?

\parhead{Cavity method.}
The \emph{cavity method}, and the related \emph{replica method}, originated in physics to predict the properties of various Gibbs distributions.
Some striking achievements of this heuristic in producing accurate predictions are the Parisi formula \cite{Par80,Tal06}, and the $k$-SAT satisfiability threshold \cite{MPZ02,MMZ06,DSS15}.
It was also employed in the work of Decelle, Krazakala, Moore, \& Zdeborov{\'a} \cite{DKMZ11} to conjecture the Kesten--Stigum threshold as the computational threshold for SBM.

Of particular interest to us are the works of Coja-Oghlan, Krzakala, Perkins \& Zdeborov{\'a} \cite{COKPZ17,COP19b,COP19}, that characterize the recovery rate that the optimal estimator, namely sampling from the Gibbs distribution, achieves for various planted constraint satisfaction problems.
Their proofs use fairly minimal properties of the Gibbs distribution.
More concretely, for a graph $G$ and an assignment $x$, let $H_G(x)$ be a Hamiltonian, and let $\pi_G$ be the corresponding Gibbs distribution.
Their proofs rely on the following properties satisfied by Gibbs distributions.
\begin{itemize}
    \item (Gibbs ratios) For any graph $G$ and any vertex $v$:
    \[
        \frac{\pi_{G}(x)}{\pi_{G \setminus v}(x)} \propto \exp(H_G(x) - H_{G\setminus v}(x)),
    \]
    \item (Approximate independence) For random $\bG$ and $v$, the marginals of the neighbors of $v$ are approximately independent in $\pi_{\bG\setminus v}$.
\end{itemize}

If one can show that a family of locally stationary distributions also satisfy the Gibbs ratios up to a multiplicative error and approximate independence on a random graph $\bG$, then one could hope to port over the cavity method predictions and their rigorous proofs in a black-box fashion.
\begin{problem}[Cavity method]
    Let $\{\nu_G\}_{G\in\text{graphs}}$ be a family of locally stationary distributions where $\nu_G$ arises from running the Glauber dynamics for $\pi_G$ for time-$T$ initialized at the uniform distribution.
    For a random graph $\bG$ and random vertex $\bv$, do $\nu_{\bG}$ and $\nu_{\bG\setminus \bv}$ satisfy, up to small error, the Gibbs ratios and approximate independence properties?
\end{problem}

\parhead{Beyond average-case models.}
Our work proves that Glauber dynamics recovers planted spikes when the input matrix has a clean ``signal + noise'' structure.
Recently, there has been a flurry of work on inference in \emph{semirandom models}; see, e.g., \cite{BKS23,GHKM23} and the references within, where it is possible to extract the hidden signal using semidefinite programming-based algorithms.
A natural direction is to investigate whether Glauber dynamics solves semirandom inference problems.
\begin{problem}[Semirandom models]
    Does Glauber dynamics succeed at finding solutions to semirandom planted CSPs, or large cliques in semirandom graphs with planted cliques as is done in the works of \cite{GHKM23} and \cite{BKS23} respectively?
\end{problem}
In a similar vein, semidefinite programming has been phenomenally successful at solving dense CSPs, and more generally CSPs on graphs with low threshold-rank \cite{BRS11}.
A reason to believe that local algorithms perform well is \Cref{th:correlation-gain}, where we show that Glauber dynamics can recover rank-1 spikes in threshold-rank-$1$ matrices.
\begin{problem}[CSPs on low threshold-rank graphs]
    Does running Glauber dynamics give a \textsf{PTAS} for Max Cut on a dense graph?
\end{problem}


\section{Preliminaries}\label{sec:prelims}
We begin by setting up some notation.
\begin{itemize}
    \item Let $P$ be the transition matrix of a time-reversible Markov chain on state space $\Omega$ with stationary distribution $\pi$, where $P[i,j]$ denotes the transition probability from $i$ to $j$.
    Let $P_t = \exp(-t(I - P))$ denote the time-$t$ transition kernel.
    \item For a distribution $\nu$ absolutely continuous with respect to $\pi$, we use $\density(x) \coloneqq \frac{\dif\nu}{\dif\pi}(x)$ to refer to its relative density to $\pi$. 
    \item We use $\nu_t$ to denote $P_t \nu$, and we assume that $\nu_t$ is absolutely continuous with respect to $\pi$ throughout. In particular, we write $\density_t(x) \coloneqq \frac{\dif \nu_t}{\dif \pi}(x)$ to denote its relative density to $\pi$; 
   \item  We use $\mean(\nu) \coloneqq \E_{\bx \sim \nu} \bx$ to denote the mean of $\nu$.
    \item We use $\by\sim_P x$ when $\by$ is chosen as a random neighbor of $x$ according to transition probabilities given by $P$.
    We drop the subscript $P$ from the $\sim$ when the Markov chain is clear from context.
\end{itemize}

\begin{remark}
    The way to think of the time-$t$ transition kernel $P_t$ for a Markov chain with kernel $P$ on a discrete space is via the process: sample $\bt\sim\Poi(t)$ and take $\bt$ steps using $P$.
\end{remark}

We will require the following simple consequence  of the definition of total variation distance.

\begin{fact}
    \label{lem:TV-stability}
    For any pair of distributions $\nu$ and $\pi$ on $\Omega$, and any function $f:\Omega\to\R$, we have
    \[
        \left|\E_{\nu} f - \E_{\pi} f \right| \le \parens*{f_{\max} - f_{\min}}\cdot \dtv{\nu}{\pi}.
    \]
\end{fact}

\begin{definition}[Dirichlet form]
    For
    functions $f,g:\Omega\to\R$, and $x,y\in\Omega$, the \emph{Dirichlet form} of $f$ and $g$ with respect to $P$ is:
    \[
        \calE_P\parens*{f,g} \coloneqq \E_{\bx\sim\pi} \E_{\by\sim_P {\bx}} (f(\bx)-f(\by))\cdot(g(\bx)-g(\by))\mper
    \]
    We drop the $P$ in the subscript when it is clear from context.
\end{definition}

\begin{remark}
    When we use the Glauber dynamics chain for a distribution $\pi$ on a hypercube, we use $\calE_{\pi}$ to denote the corresponding Dirichlet form.
\end{remark}

The Dirichlet form measures the rate at which a Markov chain makes progress towards the stationary distribution.
The following is one way of articulating this notion; see, e.g., \cite{BT06}.
\begin{fact} \label{fact:derivative-KL}
    $\frac{\dif}{\dif t}\KL{\nu_t}{\pi} = -\calE(\density_t,\log\density_t) = -\E_{\bx\sim \pi} \E_{\by\sim\bx}\parens*{\density_t(\bx) - \density_t(\by)} \cdot \log\frac{\density_t(\bx)}{\density_t(\by)}$.
\end{fact}

\begin{definition}[Modified log-Sobolev inequality]\label{def:mlsi}
    We say
    $P$ satisfies
    a \emph{modified log-Sobolev inequality} (MLSI) with constant $C$ if for any function $f : \Omega \to \R_{> 0}$,
        \[ \mathcal{E}_{P}(f,\log f) \ge C \cdot \Ent[f]. \]
    Here, $\Ent[f] \coloneqq \E_{\pi}[f \log f] - \E_{\pi} f \log \E_{\pi} f$ is the entropy functional.
    In particular, $\MLSI$ is the best (largest) such constant $C$.
\end{definition}

We will need the following fact concerning the MLSI for product measures.
\begin{fact}[{see, e.g. \cite[Lemma 2.5]{Goe04}}]\label{fact:mlsi-product}
Let $\pi$ be a distribution over $\{\pm 1\}^{n}$ with independent coordinates. Then $\MLSI(\pi) \geq 1/n$.
\end{fact}

Finally, we need the following lemma showing that an MLSI implies concentration of Lipschitz functions, in particular linear forms.
\begin{lemma}[{\cite[Theorem 5.1]{Goe04}}]
    \label{lem:mlsi-to-concentration}
    Let $\mu$ be an arbitrary distribution over $\{\pm 1\}^n$ such that $\MLSI(\mu) \ge \alpha$. Then, for any function $f : \{\pm 1\}^n \to \R$ that is $1$-Lipschitz, in that $|f(x) - f(y)| \le 1$ if $\|x-y\|_1 \le 2$,
    \[
        \Pr_{\bx \sim \mu} \left[ \left| f(\bx) - \E_{\mu} f \right| \ge t \right] \le 2 \exp\left( - \frac{\alpha t^2}{2} \right).
    \] 
\end{lemma}

\begin{remark}
    It is well known that the KL divergence to the stationary distribution decays exponentially at a rate dictated by $\MLSI$ (see {\cite[Theorem 2.4]{BT06}} for more details):
    \[
        \KL{\nu_t}{\pi} \le \KL{\nu_0}{\pi} \cdot \exp\left( - \MLSI t \right).
    \]
\end{remark}

\parhead{Measure decompositions.}
Some of the key properties of locally stationary distributions rely on the notion of a measure decomposition. 
These can be defined in great generality, but we will restrict our attention to distributions on subsets of $\R^n$ for concreteness.
\begin{definition}[Measure decomposition]
    Let $\pi$ be a distribution on $\R^n$. Let $\rho$ be a mixture distribution, also on $\R^n$, which indexes into a family of mixture components $\{\pi_z\}_{z \in \R^n}$. 
    We say that $(\rho, \pi_z)$ is a measure decomposition for $\pi$ if 
    \[
        \pi = \E_{\bz\sim\rho} \pi_{\bz}\mper
    \]
\end{definition}
One should think of the mixture components $\pi_z$ as being ``simpler'' distributions than the original measure $\pi$.
Not all measure decompositions are useful; there is always a trivial measure decomposition where the mixture $\rho$ is exactly $\pi$ and the simpler distributions $\pi_{z}$ are just Dirac masses at $z$.

Associated to each measure decomposition is a natural Markov chain; see, e.g., \cite[Definition 6]{CE22}.
\begin{definition}[Markov chain associated to a measure decomposition]
    Given a measure decomposition $\pi = \E_{\bz \sim \rho} \pi_{\bz}$, its associated Markov chain is defined by 
    \begin{align*}
        \bx \to \bz | \bx \to \bx' | \bz.
    \end{align*}
\end{definition}
Notably, Glauber dynamics and restricted Gaussian dynamics can both be viewed as the associated Markov chain to certain measure decompositions.
The relevant decomposition for Glauber dynamics represents $\pi$ as the mixture of its conditional marginals.
\begin{remark}
    Measure decompositions constructed using stochastic localization have recently been used to prove functional inequalities for a wide class of Ising models \cite{EKZ22,CE22,LMRW}. 
    We will discuss these properties in greater depth in \Cref{sec:ising-prelims}.
\end{remark} 

\parhead{Symmetric KL divergence.}
A useful potential function for us is the symmetric KL divergence.
\begin{definition}
    For a pair of distributions $\pi$ and $\nu$ on $\Omega$, we define their \emph{symmetric KL divergence} as:
    \[
        \SKL{\pi}{\nu} \coloneqq \KL{\nu}{\pi} + \KL{\pi}{\nu}.
    \]
\end{definition}

\begin{observation} \label{obs:symKL-dirichlet}
    For any $\pi$ and $\nu$, setting $f$ to be the density of $\nu$ with respect to $\pi$, we have
    \[
        \SKL{\pi}{\nu} = \frac{1}{2} \cdot \E_{\bx,\by\sim\pi} \left[\parens*{f(\bx)-f(\by)}\cdot\log\frac{f(\bx)}{f(\by)}\right]\mper
    \]
    Observe that the above quantity is the Dirichlet form for the trivial ``one-step'' Markov chain with transition matrix $\bone \pi^{\top}$ with stationary distribution $\pi$, which we shall denote $K(\pi)$.
\end{observation}

\begin{lemma}   \label{lem:symKL-to-KL}
    Let $\nu$ be an arbitrary distribution with density $\density$ with respect to $\pi$, and $\tau$ such that $\tau > \max_{x \in \Omega} \log \density(x)$ or $\tau > \max_{x \in \Omega} \log \frac{1}{\density(x)}$. Then, the symmetric KL divergence can be bounded in terms of the KL divergence as follows.
    \begin{align*}
        \SKL{\pi}{\nu} \le \parens*{ 6 + 12 \tau } \cdot \KL{\nu}{\pi} \mper
    \end{align*}
\end{lemma}
\begin{proof}
    Recall that the Hellinger distance is defined as
    \[ \Hel{\nu}{\pi} \defeq \frac{1}{2} \cdot \E_{\substack{\bx \sim \pi \\ \by \sim \pi}} \left[ \left(\sqrt{\density(\by)} - \sqrt{\density(\bx)}\right)^2 \right]. \]
    The symmetric KL divergence may be bounded by the Hellinger distance as follows.
    \begin{align*}
        2 \cdot \SKL{\nu}{\pi} &= \E_{\substack{\bx\sim\pi \\ \by \sim \pi}} \left[ \left(\density(\by)-\density(\bx)\right) \left(\log \density(\by) - \log \density(\bx)\right) \right] \\
        &= \E_{\substack{\bx\sim\pi \\ \by \sim \pi}} \left[ \left(\density(\by)-\density(\bx)\right) \cdot \log \frac{\density(\by)}{\density(\bx)} \cdot \Ind\bracks*{ \frac{\density(\by)}{\density(\bx)} \in \left[ \frac{1}{2} , 2 \right] } \right] \\
        &\qquad\qquad\qquad\qquad+ \E_{\substack{\bx\sim\pi \\ \by \sim \pi}} \left[ \left(\density(\by)-\density(\bx)\right) \cdot \log \frac{\density(\by)}{\density(\bx)} \cdot \Ind\bracks*{ \frac{\density(\by)}{\density(\bx)} \not\in \left[ \frac{1}{2} , 2 \right] } \right] \\
        &= \E_{\substack{\bx\sim\pi \\ \by \sim \pi}} \left[ \left(\density(\by)-\density(\bx)\right) \cdot \log \frac{\density(\by)}{\density(\bx)} \cdot \Ind\bracks*{ \frac{\density(\by)}{\density(\bx)} \in \left[ \frac{1}{2} , 2 \right] } \right] \\
        &\qquad\qquad\qquad\qquad+ 2 \cdot \E_{\substack{\bx\sim\pi \\ \by \sim \pi}} \left[ \left(\density(\by)-\density(\bx)\right) \cdot \log \density(\by) \cdot \Ind\bracks*{ \frac{\density(\by)}{\density(\bx)} \not\in \left[ \frac{1}{2} , 2 \right] } \right].
    \end{align*}
    To control the first term, we note that for $t \in \left[ \frac{1}{2} , 2 \right]$, $(t-1)\log t \le 6\left(\sqrt{t} - 1\right)^2$. Consequently,
    \begin{align*}
        &\E_{\substack{\bx \sim \pi \\ \by \sim \pi}} \left[ \left(\density(\by)-\density(\bx)\right) \cdot \log \frac{\density(\by)}{\density(\bx)} \cdot \Ind\bracks*{ \frac{\density(\by)}{\density(\bx)} \in \left[ \frac{1}{2} , 2 \right] } \right] \\
        &\qquad\qquad\le 6 \cdot \E_{\substack{\bx\sim\pi \\ \by \sim \pi}} \left[ \left(\sqrt{\density(\by)} - \sqrt{\density(\bx)} \right)^2 \cdot \Ind\bracks*{\frac{\density(\by)}{\density(\bx)} \in \left[\frac{1}{2},2\right] } \right] \\
            &\qquad\qquad\le 12 \Hel{\nu}{\pi}.
    \end{align*}
    To control the second term, we first have
    \[ \E_{\substack{\bx\sim\pi \\ \by \sim \pi}} \left[ \left(\density(\by)-\density(\bx)\right) \cdot \log \density(\by) \cdot \Ind\bracks*{ \frac{\density(\by)}{\density(\bx)} \not\in \left[ \frac{1}{2} , 2 \right] } \right] \le \tau \cdot \E_{\substack{\bx\sim\pi \\ \by \sim \pi}} \left[ \left|\density(\by)-\density(\bx)\right| \cdot \Ind\bracks*{ \frac{\density(\by)}{\density(\bx)} \not\in \left[ \frac{1}{2} , 2 \right] } \right] \]
    Furthermore, for $t \not\in \left[\frac{1}{2} , {2}\right]$, $|t-1| \le 6(\sqrt{t}-1)^2$. Consequently,
    \begin{align*}
        \E_{\substack{\bx\sim\pi \\ \by \sim \pi}} \left[ \left|\density(\by)-\density(\bx)\right| \cdot \Ind\bracks*{ \frac{\density(\by)}{\density(\bx)} \not\in \left[ \frac{1}{2} , 2 \right] } \right] &\le 6 \cdot \E_{\substack{\bx\sim\pi \\ \by \sim \pi}} \left[ \left(\sqrt{\density(\by)} - \sqrt{\density(\bx)} \right)^2 \cdot \Ind\bracks*{ \frac{\density(\by)}{\density(\bx)} \not\in \left[ \frac{1}{2} , 2 \right] } \right] \\
        &\le 12 \Hel{\nu}{\pi}.
    \end{align*}
    Putting this together, we get that
    \[ \SKL{\nu}{\pi} \le \left(6 + 12\tau\right) \cdot \Hel{\nu}{\pi}. \]
    To complete the proof, we use the well-known fact that $\Hel{\nu}{\pi} \le \KL{\nu}{\pi}$ (\cite{HW58}, also see \cite[Equation (16)]{SV16}).
\end{proof}

\section{Properties of locally stationary distributions}
\label{sec:lsm-props}

We record some useful properties of locally stationary distributions below.

\parhead{Random walks yield locally-stationary measures at a typical time.}
The following is a generic statement about any time-reversible Markov chain achieving a locally stationary distribution.
\begin{lemma}   \label{lem:QS}
    For any distribution $\nu$, any Markov chain transition kernel $P$ with stationary distribution $\pi$ and any $T > 0$, for $t$ chosen uniformly at random in $[0,T]$:
    \[
        \E_{\bt\sim[0,T]} \calE_P(\density_{\bt}, \log\density_{\bt}) \le \frac{\KL{\nu}{\pi}}{T} \le \frac{\log\frac{1}{\pi_{\min}}}{T}.    \numberthis \label{eq:QS}
    \]
\end{lemma}
\begin{proof}
    By \Cref{fact:derivative-KL},
    \begin{align*}
        0 &\le \KL{\nu_T}{\pi} \\
        &= \KL{\nu}{\pi} - \int_{0}^T \calE(\density_t,\log\density_t) \dif t \\
        &= \KL{\nu}{\pi} - T \cdot \E_{\bt \sim [0,T]} \calE(\density_{\bt}, \log \density_{\bt}) \\
        &\le \log\frac{1}{\pi_{\min}} - T\cdot \E_{\bt\sim[0,T]} \calE(\density_{\bt}, \log\density_{\bt}).
    \end{align*}
    Rearranging the above gives us the desired statement.
\end{proof}

A simple consequence of \pref{lem:QS} and Markov's inequality is that for most times in $[0,T]$, $\nu_t$ is indeed locally stationary.

\mainls*

\paragraph{Stationarity over small time-scales.}

We will also require the observation that if the Dirichlet form at a distribution is small, so too is the total variation distance between it and the distribution obtained after one step of the Markov chain.

\begin{lemma}
    \label{lem:dirichlet-tv-relation}
    Let $P$ be a reversible Markov chain with stationary distribution $\pi$,  $\nu$ an arbitrary distribution, and $\density$ its density relative to $\pi$. Then,
    \[ \calE_{P}(\density, \log \density) \ge 2 \cdot \KL{P\nu}{\nu} \ge 4 \cdot \dtv{\nu}{P\nu}^2. \]
\end{lemma}
\begin{proof}
    We recall the definition of the Dirichlet form,
    \begin{align*}
        \calE_{P}(\density, \log \density) &= \E_{\substack{\bx \sim \pi \\ \by \sim_P \bx}} (\density(\bx)-\density(\by))\log\frac{\density(\bx)}{\density(\by)} \\
            &= 2 \cdot \E_{\substack{\bx \sim \pi \\ \by \sim_P \bx}} \density(\bx) \log \left( \frac{\density(\bx)}{\density(\by)} \right) \tag{Reversibility of $P$}\\
            &= 2 \cdot \E_{\substack{\bx \sim \nu \\ \by \sim_P \bx}} \log \left( \frac{\density(\bx)}{\density(\by)} \right) \\
            &= 2 \cdot \left(\E_{\bx \sim \nu} \log \density(\bx) - \E_{\by \sim P\nu} \log \density(\by)\right).
    \end{align*}
    Adding and subtracting a term, we may neatly express the above in terms of KL divergences as
    \begin{align*}
        \calE_{P}(\density, \log \density) &= 2 \cdot \left[ \left( 
\E_{\bx \sim \nu} \log \density(\bx) - \E_{\by \sim P\nu} \log \frac{\dif P\nu}{\dif \pi}(\by) \right) + \left( \E_{\by \sim P\nu} \log \frac{\dif P\nu}{\dif \pi}(\by) - \E_{\by \sim P\nu} \log \density(\by) \right) \right] \\
            &= 2 \cdot \left[ \left( \KL{\nu}{\pi} - \KL{P\nu}{\pi} \right) + \KL{P\nu}{\nu} \right] \\
            &\ge 2 \cdot \KL{P\nu}{\nu} \ge 4 \cdot \dtv{P\nu}{\nu}^2
    \end{align*}
     as desired, where the second-to-last inequality follows from the fact that the KL divergence to the stationary distribution is non-increasing with time, and the last inequality is Pinsker's.
\end{proof}

Consequently, averages of bounded functions do not change much after one step of the Markov chain
\begin{corollary}   \label{cor:bounded-function-stability}
    Let $\phi: \Omega \to \R$ be a bounded function on the state space of a Markov chain $P$, and $\nu$ be an $\eps$-locally stationary measure.  Then,
    \[
        | \E_{\bx \sim \nu}[ \phi(\bx)] - \E_{\bx \sim P \nu} [ \phi(\bx)] | \leq \norm*{\phi}_{\infty} \cdot \sqrt{\eps}.
    \]
\end{corollary}

\parhead{Locally stationary measures are close to stationary measures on small neighborhoods.}
As their name suggests, locally stationary distributions locally resemble the true stationary distribution.
For example, typical samples from the locally stationary distribution approximately satisfies the detailed balance condition.
Even though we do not explicitly employ this in any applications, we include it here as it gives the impression of a fundamental structural property of locally stationary distributions.
\begin{lemma}   \label{lem:if-not-log-sob}
    For an $\eps$-locally stationary distribution $\nu$ with relative density $\density$, and for $\bx\sim\nu$ and $\by\sim\bx$, with probability at least $1-\delta$, we have
    \(
        \frac{\density(\bx)}{\density(\by)} = 1 \pm O\parens*{\sqrt{\frac{\eps}{\delta}}},
    \)
    where $\delta > 2\eps$.
\end{lemma}
\begin{proof}
    Since $\calE(\density,\log\density) < \eps$,
    \begin{align*}
        \E_{\bx\sim \pi} \E_{\by\sim\bx}\parens*{\density(\bx) - \density(\by)} \cdot \log\frac{\density(\bx)}{\density(\by)} &< \eps \\
        \E_{\bx\sim\nu} \E_{\by\sim\bx} \parens*{1 - \frac{\density(\by)}{\density(\bx)}} \cdot \log\frac{\density(\bx)}{\density(\by)} &< \eps
    \end{align*}
    Since the random variable at hand is always nonnegative, we can apply Markov's inequality, which tells us that with probability at least $1-\delta$:
    \[
        \parens*{1 - \frac{\density(\by)}{\density(\bx)}} \cdot \log\frac{\density(\bx)}{\density(\by)} < \frac{\eps}{\delta}.
    \]
    The claim then follows since the above inequality is violated if $\frac{\density(\bx)}{\density(\by)}$ deviates from $1$ by more than a constant multiple of $\sqrt{\frac{\eps}{\delta}}$.
\end{proof}

Formally, the following can be abstracted out of the proof of \Cref{thm:glauber-indepset}:
\begin{lemma} \label{lem:localpatches}
Let $P$ be a Glauber dynamics chain for a distribution $\mu$ on $\{\pm1\}^n$, and let $\nu$ be an $\eps$-locally stationary measure with respect to $P$ for some $\eps > 0$.
For a subset of coordinates $W \subset [n]$, and an assignment $x_{\ol{W}}$ of coordinates outside $W$, let $P_{W,x_{\ol{W}}}$ denote the Glauber dynamics chain of $\mu|x_{\ol{W}}$.
Suppose for every choice of $W$ and $x_{\ol{W}}$, we have $\MLSI\parens*{P_{W,x_{\ol{W}}}} \ge C$, then:

$$ \E_{\bx_{\ol{W}} \sim \nu} \left[\dtv{\nu_{|\bx_{\ol{W}}}}{\pi_{|\bx_{\ol{W}}}} \right] \leq \frac{1}{C} \cdot \sqrt{\eps} \mper $$
\end{lemma}

\begin{corollary}
In the setting of \prettyref{lem:localpatches}, suppose $\phi : \{0,1\}^W \to \R$ is a bounded function of $x_W$ then,
\begin{align*}
    \E_{\bx \sim \nu}\bracks*{\phi\parens*{x_W}} \ge \E_{\bx_{\ol{W}}} \E_{\bx_W\sim \pi|_{x_{\ol{W}}} } \bracks*{\phi(x_W)} - \frac{1}{C} \cdot \sqrt{\eps} \cdot \norm*{\phi}_{\infty}\mper
\end{align*}
\end{corollary}

\parhead{Local stationarity can be transferred under component MLSI.}\label{sec:rgd-props}
Let $\pi = \E_{\bz \sim \rho} \pi_{\bz}$ be a measure decomposition.
In this section we prove that if Glauber dynamics is locally stationary, so too is the Markov chain associated with the measure decomposition, provided that the mixture components $\pi_z$ all have a good MLSI constant. 
\begin{restatable}{lemma}{gdtorgd}   \label{lem:ls-dist-rgd}
    Let $P$ be the Markov chain associated to a measure decomposition $\pi = \E_{\bz \sim \rho} \pi_{\bz}$. 
    Let $f: \{\pm1\}^n \to \R_{>0}$ be any function and set $\tau$ such that $\min_{x \in \{\pm 1\}^n} f(x) > \exp(-\tau)$ or $\max_{x \in \{\pm 1\}^n} f(x) < \exp(\tau)$.
    For $\delta \coloneqq \inf_{z} \MLSI\parens*{\pi_{z}}$, we have
    \[
        \calE_{P}(f, \log f) \le O\parens*{\frac{\tau}{\delta}} \cdot \calE_{\pi}(f,\log f)\mper
    \]
\end{restatable}
\begin{proof}
    We use $\calC_n$ to denote the hypercube graph on vertex set $\{\pm1\}^n$ with edge set having pairs of vertices that differ in a single coordinate.
    For any nonnegative function $f$ and distribution $\pi$ with $\E_\pi f = 1$, we use $f\cdot\pi$ to denote the distribution $\nu$ with $\frac{\dif \nu}{\dif \pi}(x) = f(x)$.   

    For any function $f$ satisfying the assumption of the statement, we have
    \begin{align*}
        \calE_{\pi}(f, \log f) &= \sum_{\{x,y\}\in\calC_n} \frac{1}{n} \cdot \frac{\pi(x) \cdot \pi(y)}{\pi(x) + \pi(y)} \cdot (f(x) - f(y))\log \frac{f(x)}{f(y)}  \\
        &= \sum_{\{x,y\}\in\calC_n} \frac{1}{n}\cdot\frac{\E_{\bz\sim\rho}\pi_{\bz}(x) \cdot \E_{\bz\sim\rho}\pi_{\bz}(y)}{\E_{\bz\sim\rho}\pi_{\bz}(x) + \E_{\bz\sim\rho}\pi_{\bz}(y)} \cdot (f(x) - f(y))\log \frac{f(x)}{f(y)}  \\
        &\ge \sum_{\{x,y\}\in\calC_n} \frac{1}{n}\cdot \E_{\bz\sim\rho} \left[ \frac{\pi_{\bz}(x) \cdot \pi_{\bz}(y)}{\pi_{\bz}(x)+\pi_{\bz}(y)} \right] \cdot (f(x) - f(y))\log \frac{f(x)}{f(y)} \\
        &= \E_{\bz\sim\rho} \left[ \calE_{\pi_{\bz}}(f, \log f) \right].
    \end{align*}
    Above, the inequality follows from the concavity of the function $(a,b) \mapsto \frac{ab}{a+b}$ in the non-negative quadrant, and all the Dirichlet forms are with respect to the Glauber dynamics chain.
    
    Continuing the above calculation,
    \begin{align*}
        \E_{\bz \sim \rho} [\calE_{\pi_{\bz}}(f, \log f)] &\ge \E_{\bz\sim\rho} \left[ \MLSI(\pi_{\bz}) \cdot \E_{\pi_{\bz}}[f] \cdot \KL{\frac{f}{\E_{\pi_{\bz}}f} \cdot \pi_{\bz}}{\pi_{\bz}} \right] \\
        &\ge \Omega\parens*{ \frac{\delta}{\tau} } \E_{\bz\sim\rho} \left[ \E_{\pi_{\bz}}[f]  \cdot \SKL{\frac{f}{\E_{\pi_{\bz}} f} \cdot \pi_{\bz}}{\pi_{\bz}} \right] & \text{(by \Cref{lem:symKL-to-KL} and MLSI)} \\
        &= \Omega\parens*{\frac{\delta}{\tau}} \E_{\bz\sim\rho} \left[ \E_{\bx,\by \sim \pi_{\bz}} (f(\bx) - f(\by)) \log \frac{f(\bx)}{f(\by)} \right] &\text{(by \Cref{obs:symKL-dirichlet})} \\
        &= \Omega\parens*{ \frac{\delta}{\tau} } \calE_{P}(f, \log f)\mper
    \end{align*}
    The claim follows.
\end{proof}
The upshot is that we have complete freedom to select the measure decomposition, provided we can establish an MLSI for the components. 
This can be useful when it is easier to directly analyze the consequences of local stationarity for the associated Markov chain instead of Glauber dynamics. 


\section{Warmup: Large independent sets in triangle-free graphs}\label{sec:glauber-indepset}
Observe that any graph $G$ on $n$ vertices with maximum degree $d$ has an independent set of size $\frac{n}{d+1}$, a bound which is tight for the disjoint union of $(d+1)$-sized cliques.
Ajtai, Koml{\'o}s, and Szemer{\'e}di \cite{AKS80} showed that when $G$ is triangle-free, the size of the maximum independent set increases to $\Omega\parens*{n\cdot\frac{\log d}{d}}$.
Shearer \cite{She83} gave an alternate proof that shows such an independent set exists with a leading constant of $1$, even if $G$ merely has average-degree bounded by $d$. As a warmup, we prove that Glauber dynamics succeeds at finding a large independent set in $O(nd^{4})$ steps.

\glauberindepset*


%

To prove \cref{thm:glauber-indepset}, we will need the following crude bound on the modified log-Sobolev constant for the uniform distribution over independent sets of a star. A short proof is provided at the end of this section.

\begin{lemma}\label{lem:mlsi-indepset-star}
Let $\pi$ denote the uniform distribution over independent sets of a star with $\Delta$ many leaves. Then $\MLSI(\pi) \geq \exp(-O(\Delta))$.
\end{lemma}

\begin{remark}
The bound can easily be made $\frac{1}{\poly(\Delta)}$, but we will not need this here.
\end{remark}

We also leverage the following simple and well-known lemma on the local behavior of a uniformly random independent set. For completeness, we include a short proof of it at the end of this section, following the one provided in Alon \& Spencer \cite[Proposition 1, Page 272]{AS16}. Throughout this section, we write $N(v) = \{u \in V : u \sim v\}$ for the open neighborhood of $v \in V$, and $N[v] = N(v) \cup \{v\}$ for the closed neighborhood.
\begin{lemma}\label{lem:scorefunc-lb}
Let $G$ be a triangle-free graph of maximum degree $d$, and let $\pi$ denote the uniform measure over independent sets of $G$. For every vertex $v \in V$, define the following real-valued score function over $\{0,1\}^{V}$:
\begin{align}\label{eq:indep-scorefunc}
    \phi_v(\bx) \coloneqq d \bx_v + \sum_{u\in N(v)} \bx_u.
\end{align}
Then for every pinning $\tau \in \{0,1\}^{\ol{N[v]}}$, we have
$$ \E_{\bx \sim \pi}[ \phi_{v}(\bx) \mid \bx_{\ol{N[v]}} = \tau] \geq \frac{\log d}{2}.$$
\end{lemma}

The key property of this score function is that it readily yields a lower bound on the size of an independent set $x \in \{0,1\}^{V}$. This follows from the observation that
\begin{align}\label{eq:scorefunc-indepsetsize}
    n\cdot\E_{v\sim V} \phi_v(\bx) \le 2d\cdot\sum_{v \in V} \bx_v.
\end{align}
Note that by averaging over $\tau \in \{0,1\}^{\ol{N[v]}}$ drawn from the marginal distribution of $\pi$ induced on $\ol{N[v]}$, the conclusion of \cref{lem:scorefunc-lb} combined with \cref{eq:scorefunc-indepsetsize} implies that a uniformly random independent set drawn from $\pi$ has expected size at least $\frac{1}{4} \cdot n \cdot \frac{\log d}{d}$.
We observe that the same claim holds even if the distribution over independent sets is merely locally stationary with respect to Glauber dynamics, rather than being truly uniform.

\begin{proof}[Proof of \cref{thm:glauber-indepset}]
    As discussed above, a direct application of the law of total expectation combined with \cref{lem:scorefunc-lb} yields the first claim concerning the expected size of a uniformly random independent set. We now turn to the second claim. Let $T \geq 0$ be a parameter to be determined later, and for every $0 \leq t \leq T$, let $\nu_t$ denote the distribution over independent sets after running Glauber dynamics for time-$t$ from an arbitrary initialization. Our goal is to establish the lower bound
    \begin{align}\label{eq:glauber-Escore-lb}
        \E_{t \sim [0,T]} \E_{\nu_t} \E_{v\sim V} \phi_v(\bx) \ge \frac{\log d}{2} - \eps
    \end{align}
    for $0 < \eps < o_{d}(1)$, which when combined with \cref{eq:scorefunc-indepsetsize} immediately implies that the expected size of the independent set discovered by Glauber dynamics is $\parens*{\frac{1-o_d(1)}{4}}\cdot n\cdot\frac{\log d}{d}$. For the purpose of analysis, if $\nu$ is any distribution over independent sets, we shall think of $\bx \sim \nu$ as being sampled in the following alternate way.
    \begin{enumerate}
        \item For a fixed vertex $v$, sample $\bx_{\ol{N[v]}}$ from the marginal distribution induced by $\nu$ on $\ol{N[v]}$. For each $w \in N(v)$ that has a neighbor in the independent set $\bx_{\ol{N[v]}}$, pin $\bx_w$ to $0$, since it is deterministically equal to $0$ in the conditional measure $\nu|\bx_{\ol{N[v]}}$. 
        \item If the number of unpinned vertices at this stage is strictly larger than $\log d$, sample $\bx_v$ from its corresponding conditional marginal distribution.
        \item Let $U$ be the set of remaining unpinned vertices.
        Sample $\bx_U\sim\nu|\bx_{\ol{U}}$.
    \end{enumerate}

    For any vertex $v\in V$, we have
    \begin{align*}
        \E_{\bx\sim\nu} \phi_v(\bx) &= \E_{U,\bx_{\ol{U}}|v} \E_{\nu|\bx_{\ol{U}}} \phi_v(\bx) \\
        &\ge \E_{U,\bx_{\ol{U}}|v} \bracks*{\E_{\pi|\bx_{\ol{U}}} \phi_v(\bx) - 2d\cdot\dtv{\nu|\bx_{\ol{U}}}{\pi|\bx_{\ol{U}}}} \\ 
        &\ge \frac{\log d}{2} - 2d\cdot \E_{U,\bx_{\ol{U}}|v}\dtv{\nu|\bx_{\ol{U}}}{\pi|\bx_{\ol{U}}}, \numberthis \label{eq:potential-tv}
    \end{align*}
    Note that the random subset of vertices $U$, as well as the boundary condition $\bx_{\ol{U}}$, are all drawn from the above process with respect to $\nu$, not $\pi$. The first inequality follows by applying \cref{lem:TV-stability} along with $2d$-boundedness and nonnegativity of the score function $\phi_{v}$. For the second inequality, note that if $v \in U$, then we may invoke \cref{lem:scorefunc-lb}. Now suppose $v \notin U$. If $v$ is pinned $1$, then $\phi_{v}(\bx) = d$. If $v$ is pinned $0$, then by triangle-freeness, $\bx_{u} = 1$ with probability $1/2$ independently for all $u \in U$. Since $|U| \geq \log d$, the lower bound follows.
 
    In the rest of this argument, we will show that when $\bt\sim[0,T]$, $\nu$ is equal to $\nu_{\bt}$ and $v$ is chosen uniformly at random, we can achieve a strong upper bound on
    \[
        \E_{t} \E_v \E_{U,\bx_{\ol{U}}|v} \dtv{\nu_t|\bx_{\ol{U}}}{\pi|\bx_{\ol{U}}}\mper
    \]
    For the rest of this proof, we shall abbreviate $\nu_t|\bx_{\ol{U}}$ and $\pi|\bx_{\ol{U}}$ as $\nu_t'$ and $\pi'$, respectively. Furthermore, let $\density_t'$ denote the relative density of $\nu_t'$ with respect to $\pi'$.
    By Pinsker's inequality, we can bound the above by:
    \begin{align*}
        \E_{t} \E_{v} \E_{U,\bx_{\ol{U}}|v} \sqrt{\KL{\nu_t'}{\pi'}} &\le \sqrt{\E_{t} \E_{v} \E_{U,\bx_{\ol{U}}|v} \KL{\nu_t'}{\pi'} }\mper
    \end{align*}
    We focus our attention on showing an upper bound on the term in the square-root.
    \begin{align*}
        \E_t\E_{v\sim V} \E_{U,\bx_{\ol{U}}|v} \KL{\nu_t'}{\pi'} &\le \E_t \E_{v} \E_{U,\bx_{\ol{U}}|v} \frac{1}{\MLSI(\pi')} \calE_{\pi'}\parens*{\density_t', \log\density_t'} \\
        &\le O(d) \cdot \E_t \E_{v} \E_{U,\bx_{\ol{U}}|v} \calE_{\pi'}\parens*{\density_t', \log\density_t'} \\
        &\le O(d^2) \cdot \E_t \calE_{\pi}\parens*{\density_t, \log\density_t} \\
        &\le O(d^2) \cdot \frac{n}{T}\mper
    \end{align*}
    In the above, the first inequality uses the definition of $\MLSI(\pi')$. For the second inequality, note that almost surely, either $G[U]$ is a star centered at $v$ with at most $\log d$ many leaves, or $G[U]$ consists entirely of isolated vertices due to pinning $v$. In either case, we have $\MLSI(\pi') \ge \Omega\parens*{\frac{1}{d}}$ by appropriately applying \cref{fact:mlsi-product} or \cref{lem:mlsi-indepset-star}. The third inequality is based on comparing Dirichlet forms.
    The final inequality is a direct application of \Cref{lem:QS}. 
    Plugging in the above into \pref{eq:potential-tv} and setting $T = O(nd^4/\eps^2)$, we get \cref{eq:glauber-Escore-lb} as desired.
\end{proof}

\begin{remark}
    Using a similar argument, one can establish a similar result for $\maxcut$ on triangle-free graphs with maximum degree $d$. In particular, Glauber dynamics run for $\poly(n)$ many steps on the antiferromagnetic Ising model on $G$ with inverse temperature $\frac{1}{\sqrt{d}}$ finds a cut of size $\frac{1}{2}+\Omega\parens*{\frac{1}{\sqrt{d}}}$.
\end{remark}



\begin{proof}[Proof of \cref{lem:scorefunc-lb}]
    Let $S \subseteq N(v)$ denote the subset of neighbors of $v$ which are not adjacent to any vertex of the independent set $\bx_{\ol{N[v]}} = \tau$, and write $k = |S|$. Observe that the distribution of $\bx_{N[v]}$ conditioned on $\bx_{\ol{N[v]}} = \tau$ is given by choosing the singleton $\{v\}$ with probability $\frac{1}{2^{k} + 1}$, or a uniformly random subset of $S$ with the remaining probability. Hence,
    \begin{align*}
        \E_{\bx \sim \pi}[ \phi_{v}(\bx) \mid \bx_{\ol{N[v]}} = \tau] = \frac{d}{2^{k} + 1} + \frac{k}{2} \cdot \frac{2^{k}}{2^{k} + 1}.
    \end{align*}
    The above expression is always at least $\frac{\log d}{2}$ for any choice of nonnegative integer $k$.
\end{proof}

\begin{proof}[Proof of \cref{lem:mlsi-indepset-star}]
By \cite[Corollary A.4]{DS96} and \cite[Proposition 3.6]{BT06}, we have that $\MLSI(\pi) \geq \frac{1 - 2\pi_{*}}{\log\parens*{\frac{1}{\pi_{*}} - 1}} \cdot \lambda(\pi)$, where $\lambda(\pi)$ denotes the spectral gap of Glauber dynamics for $\pi$, and $\pi_{*} = \min_{\bx : \pi(\bx) > 0} \pi(\bx) = \frac{1}{2^{\Delta} + 1}$. Hence, it suffices to establish that $\lambda(\pi) \geq \exp(-O(\Delta))$. For this, we appeal to the simple fact that random walk on a connected graph with $n$ vertices has spectral gap at least $1/\poly(n)$. A comparison of Dirichlet forms between Glauber dynamics and simple random walk on the $n = 2^{\Delta} + 1$ many independent sets of $G$ yields the desired lower bound.
\end{proof}


\section{Weak recovery in spiked models}
In this section, we present our main application: using Glauber dynamics on Ising models to achieve weak recovery guarantees for spiked matrix models.
Let us recall the definition of an Ising model.
\begin{definition}[Ising model]
    Let $J \in \R^{n \times n}$ be a symmetric \emph{interaction matrix} and $h \in \R^n$ an \emph{external field}.
    The \emph{Ising model} corresponding to $J$ and $h$ is the probability distribution $\mu_{J, h}$ on $\{\pm1\}^n$, where
    \[
        \mu_{J, h}(x) \propto \exp\parens*{\frac{1}{2}x^\top J x + \langle h,x\rangle}.
    \]
    We drop the $h$ from the subscript when it is equal to $0$.
\end{definition}
We now pose the general algorithmic task that we wish to solve using Glauber dynamics in \emph{polynomial} time.
\begin{problem}\label{question:weak-recovery}
    Let $W$ be a symmetric matrix in $\R^{n\times n}$ and let $v$ be a unit vector in $\R^{n}$.
    For $M \coloneqq W + \lambda vv^{\top}$, what is the behavior of Glauber dynamics run for $\poly(n)$ many steps for the Ising model $\mu_{M}$? In particular, under what assumptions does Glauber dynamics recover a vector that is well-correlated with $v$ after $\poly(n)$ time?
\end{problem}

Our main results of this section resolves \cref{question:weak-recovery} affirmatively in the following concrete settings.
\begin{itemize}
    \item $M = W + \lambda vv^\top$, where the spectral diameter of $W$ is at most $1$, and $v \in \{\pm \frac{1}{\sqrt{n}} \}^n$.
    \item $M = A_{\bG} - \frac{d}{n} \bone \bone^\top$, the degree-centered adjacency matrix of a sparse stochastic block model.
\end{itemize}

In particular, despite a failure to mix to $\mu$, Glauber dynamics run for polynomially many steps still manages to recover nontrivial information about the planted signal $v$.
\mainspikedwig*

Similarly, we have the following result for the stochastic block model.

\mainsbm*

Our strategy to study Glauber dynamics in each of these settings is to relate its behavior to that of a different Markov chain called \emph{restricted Gaussian dynamics}, whose definition we recall below.

\begin{definition}[Restricted Gaussian dynamics]
    For $M = W+\lambda vv^{\top}$, \emph{restricted Gaussian dynamics} ($\RGD$) is a Markov chain on $\{\pm1\}^n$ where a transition from $x$ to $\by$ is given by the following:
    \begin{itemize}
        \item Sample $\bg\sim\calN(0,1)$ and define $\bz\coloneqq (\lambda  \angles*{ v, x } + \sqrt{\lambda} \bg) \cdot v$.
        \item Sample $\by$ from the Gibbs distribution $\mu_{W, \bz}$.
    \end{itemize}
\end{definition}
Note that the above definition specifies the full joint distribution of $(\bx, \bz)$, which induces a measure decomposition $\mu_{M} = \E_{\bz \sim \rho} \mu_{W, \bz}$. 
Further, RGD is the associated Markov chain for this decomposition.
\begin{remark}
    Given access to $M$, but not $W$ and $v$, it is unclear how to efficiently implement restricted Gaussian dynamics. 
    For example, for the SBM application, one has $v = \frac{1}{\sqrt{n}}\bsigma$. 
    If one could compute $v$ (or $W$), then the recovery task would already be solved.  
    Nevertheless, in light of \Cref{lem:ls-dist-rgd}, it is useful for analysis because we can relate its behavior to that of Glauber dynamics on $M$.
\end{remark}

In \Cref{sec:ising-prelims}, we will introduce some preliminaries on Ising models that will be crucial in the analysis.
Then, in \Cref{sec:weak-recovery}, we show that RGD achieves weak recovery.
Finally, in \Cref{sec:glauber-weak-recovery}, we prove the main theorems about weak recovery using Glauber dynamics, \Cref{th:correlation-gain,th:sbm-recovery}.

\subsection{Entropic stability and conservation of variance}\label{sec:ising-prelims}

Recall the definition of a measure decomposition from \Cref{sec:prelims}. 
In order to prove our weak recovery result, we will need more structural properties for our decompositions. 
We now formalize these requirements.

At a high level, we would like a measure decomposition where the components of the decomposition $\pi_{\bz}$ ``inherit'' properties of $\pi$ itself, but are also simpler at the same time.
When the individual components $\pi_{\bz}$ inherit the variance of $\pi$, the mixture is said to satisfy conservation of variance.
\begin{definition}[Conservation of variance]
    We say that an Ising model $\pi$ on $\{\pm 1\}^n$ satisfies \emph{conservation of variance} for $(\rho, \pi_{z})$ with parameter $\Cvar \in [0, 1]$, if for all functions $f: \{\pm1\}^n \to \R$, we have
    \[
        \E_{\bz \sim \rho}[\Var_{\pi_{\bz}}[f]] \ge \Cvar \cdot \Var_{\pi}[f].
    \]
\end{definition}
\begin{remark}
    By the law of total variance, for any $f: \{\pm 1\}^n \to \R$,
    \[
        \Var_{\pi}[f] = \E_{\bz \sim \rho}[\Var_{\pi_{\bz}}[f]] + \Var[\E_{\pi_{\bz}} f] \ge \E_{\bz\sim\rho} \bracks*{\Var_{\pi_{\bz}}[f]}.
    \]
    The notion of conservation of variance captures mixtures where, loosely, a reverse of the above inequality is true.
\end{remark}

To control the mean of various Ising models, we will also use another notion, called \emph{entropic stability} \cite{CE22}. 
To introduce it, we need the notion of a \emph{tilt}. For a measure $\pi$ on $\Omega \subseteq \R^n$ and vector $v \in \R^n$, we define the \emph{tilted measure} $\calT_v \pi$ on $\Omega$ by 
\[
\frac{\dif \calT_v \pi(x)}{\dif \pi(x)} \propto e^{\angles{v, x}} .
\]

\begin{definition}[Entropic stability \cite{CE22}]
    Let $\Omega \subseteq \R^n$ and $\psi: \R^n \times \R^n \to \R_{\ge 0}$. 
    For $\alpha > 0$, we say that a measure $\pi$ on $\Omega$ is $\alpha$-entropically stable with respect to $\psi$ if for all $v \in \R^n$,
    \[
        \psi(\mean(\calT_v \pi), \mean(\pi)) \le \alpha \cdot \KL{\calT_v \pi}{\pi},
    \]
    where we recall that $\mean(\pi) = \E_{\bx \sim \pi} \bx$.
    We denote by $\entstab$ the best (smallest) such $\alpha$.
\end{definition}
\begin{remark}
    It turns out that entropic stability can be used to prove the related notion of conservation of \emph{entropy} for certain measure decompositions, which in turn implies conservation of variance; see, e.g., \cite[Proof of Proposition 3.5]{BT06}.
    However, we will not dwell on this point and refer the interested reader to \cite{CE22}.
\end{remark}

We now state our special measure decompositions for the two applications. For bounded spectral diameter $W$, one can decompose the Ising model into a strongly log-concave mixture of product distributions; these types of decompositions were studied in \cite{BB19,CE22}.
See \cite[Section 5.1]{CE22} for details.
\begin{theorem}\label{th:bdd-mixture}
    Let $W$ be a symmetric matrix with $\kappa \psdle W \psdle 1-\kappa$ and $h$ be an arbitrary external field. Then there exists a (strongly log-concave) mixture $\rho$ on $\R^n$ such that 
    \[
        \mu_{W, h} = \E_{\bz \sim \rho}[\mu_{0, W\bz + h}].
    \]
    Moreover, the following properties hold:
    \begin{enumerate}[label=\normalfont{(\arabic*)}]
        \item $\mu_{W, h}$ is $\frac{1}{\kappa}$-entropically stable with respect to $(x, y) \mapsto \norm{x-y}^2$.
        \item $\Cvar\left(\rho, \mu_{0, W\bz+h}\right) \ge \exp(-1/\kappa)$.
        \item \label{item:MLSI-bdd} $\MLSI\parens*{\mu_{W,h}}\ge\frac{1}{n}\cdot\frac{1}{1-\kappa}$.
    \end{enumerate}
\end{theorem}
Recently, a subset of the authors proved that a similar decomposition exists for an Ising model associated to the stochastic block model  \cite{LMRW}.
\begin{theorem}\label{th:sbm-mixture}
    Let $(\bsigma, \bG) \sim \SBM(n, d, \lambda)$, and let $\ol{A}_{\bG} = A_{\bG} - \E[A_{\bG} | \bsigma]$ be the centered adjacency matrix. There exists some constant $\beta > 0$, a mixture distribution $\rho$ on $\R^n$, and function $g: \R^n \to \R^n$ such that for any external field $h$, 
    \[
        \mu_{\beta \ol{A}_{\bG}/\sqrt{d}, h} = \E_{\bz \sim \rho}[\mu_{H, g(\bz) + h}],
    \]
    where $H$ is an interaction matrix supported on the edges of a forest, with at most one additional cycle per connected component. 
    Moreover, the following properties hold:
    \begin{enumerate}[label=\normalfont{(\arabic*)}]
        \item \label{item:entstab-sbm} There is a positive constant $\entstab$ such that $\mu_{\beta \ol{A}_{\bG}/\sqrt{d}, h}$ is $\entstab$-entropically stable with respect to $(x, y) \mapsto \norm{I_{[n] \setminus H}(x-y)}^2$.
        \item \label{item:varcon-sbm} $\Cvar(\mu_{\beta \ol{A}_{\bG}/\sqrt{d}, h}) \ge \Omega(1)$.
        \item \label{item:MLSI-sbm} $\MLSI(\mu_{\beta\ol{A}_{\bG}/\sqrt{d}, h}) \ge \frac{1}{n^{1 + o_d(1)}}$.
        \item \label{item:high-deg-sbm} The number of non-isolated vertices in $H$ is at most $\gamma n$, where $\gamma(d) = o_d(1)$ is a sufficiently small constant that also shrinks with $d$. Here, a non-isolated vertex is a vertex with at least one distinct neighbor.
        
    \end{enumerate}
\end{theorem}
\begin{remark}
    We provide pointers for the reader interested in extracting \Cref{th:sbm-mixture} from \cite{LMRW,CE22}.
    \begin{itemize}
        \item \Cref{item:entstab-sbm} follows from \cite[Theorem 5.2]{LMRW} and \cite[Lemma 40]{CE22} (cf. \cite[Lemma 3.18]{LMRW}).
        \item For \Cref{item:varcon-sbm}, conservation of entropy follows from \cite[Theorem 5.2, Lemma 3.18, Lemma 3.29]{LMRW}.
        This, in turn, implies conservation of variance by the generic fact that $\lim_{c \to \infty} \Ent[(f+c)^2] = 2\Var[f]$ (see e.g., \cite[Proof of Proposition 3.5]{BT06}).
        \item \Cref{item:MLSI-sbm} is \cite[Theorem 5.1]{LMRW}.
        \item \Cref{item:high-deg-sbm} follows from \cite[Lemma 6.21]{LMRW}.
    \end{itemize}
\end{remark}

\subsection{Restricted Gaussian dynamics achieves weak recovery}\label{sec:weak-recovery}

In this section we prove versions of \Cref{th:correlation-gain,th:sbm-recovery} for restricted Gaussian dynamics---our goal will be to show that if the Dirichlet form with respect to the restricted Gaussian dynamics Markov chain is small, then it has non-trivial correlation with the planted vector $v$. 
Since local stationarity with respect to Glauber dynamics transfers over to local stationarity with respect to RGD (\Cref{lem:ls-dist-rgd}), this suffices to complete the proof.

Showing that RGD succeeds at correlating with the spike amounts to \Cref{lem:corr-lower-bound-exp-general}, where we show that for ``nice'' Ising models, a strong external field applied to the Ising model shows itself in its mean.
In particular, if the field is aligned with $v$, then the mean is correlated with $v$.
One can then show that for any distribution $\nu$ that is locally stationary with respect to RGD, samples from $\nu$ must already have good correlation with the planted $v$ (\Cref{lem:booster-energy}).

\begin{lemma}
    \label{lem:corr-lower-bound-exp-general}
    Let $W$ be an interaction matrix such that for any external field $h$, there is a decomposition
    \[ \mu_{W,h} = \E_{\bz \sim \rho} \left[ \mu_{H,g_h(\bz)} \right], \]
    where $H$ is an arbitrary interaction matrix supported on the edges of a graph that we also denote $H$, such that the following hold.
    \begin{enumerate}
        \item There is a positive constant $\entstab$ such that $\mu_{W,h}$ is $\entstab$-entropically stable with respect to $(x,y)\mapsto \left\| I_{[n] \setminus H} (x-y) \right\|^2$,
        \item Variance is conserved in this decomposition with constant $\Cvar$.
        \item There is a constant $\gamma < 1$ such that the number of non-isolated vertices in $H$ is at most $\gamma n$.
    \end{enumerate}
    Now, let $v \in \left\{\pm \frac{1}{\sqrt{n}} \right\}^{n}$, and set $s > 0$. Then,
    \[ \E_{\bx \sim \mu_{W,sv}} |\langle \bx,v\rangle| \ge \left( 1 - \gamma \right) \cdot \frac{\Cvar}{2} \cdot \min\left\{ s , 2 \cdot \sqrt{\frac{n}{\Cvar\entstab}} \right\} \]
\end{lemma}
\begin{remark}
    The above assumptions on $H$ may be a bit confusing. However, in the case of an interaction matrix with bounded spectral diameter, we in fact have $\gamma = 0$, so $H$ is empty.
\end{remark}
\begin{remark}
    We further remark that by \cite[Lemma 3.18]{LMRW}---a mild strengthening of \cite[Lemma 40]{CE22}---the first condition holds if for all tilts $h$, $\left\| \Cov\left( \mu_{W,h} \right)_{[n] \setminus H} \right\| \le \entstab$. In fact, only a bound on $\left\| \Cov\left( \mu_{W,tv} \right)_{[n] \setminus H} \right\|$ for ``most'' $0 \le t \le s$ is required. A technique to show the second condition, involving stochastic localization, also requires a bound on the covariance.
\end{remark}
\begin{proof}
    The proof strategy is to use the fact that:
    \[
        \E_{\bx\sim\mu_{W,sv}} \angles*{\bx, v} = \int_{0}^s \frac{\dif}{\dif t} \E_{\bx\sim\mu_{W, tv}}\angles*{\bx, v} \dif t,
    \]
    and obtain a lower bound on the correlation by showing a lower bound on the derivative.

    The following standard calculation gives us a formula for the derivative as a variance of the correlation.
    \begin{align*}
        \frac{\dif}{\dif t} \E_{\bx \sim \mu_{W,tv}} \langle \bx,v\rangle &= \frac{\dif}{\dif t} \frac{\E_{\bx \sim \mu_{W,0}} e^{t\langle \bx,v\rangle} \langle \bx,v\rangle }{\E_{\bx \sim \mu_{W,0}} e^{t\langle \bx,v\rangle}} \\
            &= \frac{\E_{\bx \sim \mu_{W,0}} e^{t\langle \bx,v\rangle} \langle \bx,v\rangle^2 }{\E_{\bx \sim \mu_{W,0}} e^{t\langle \bx,v\rangle}} - \frac{\left(\E_{\bx \sim \mu_{W,0}} e^{t\langle \bx,v\rangle} \langle \bx,v\rangle\right)^2 }{\left(\E_{\bx \sim \mu_{W,0}} e^{t\langle \bx,v\rangle}\right)^2} = \Var_{\bx \sim \mu_{W,tv}} \langle \bx,v\rangle.
    \end{align*}
    Now, consider the measure decomposition provided by the assumptions, of the form 
    \[
        \mu_{W,tv} = \E_{\bz \sim \rho} \left[ \mu_{H,\bz} \right]
    \]
    for some measure $\rho$ over external fields $z$. We have by the law of total variance that 
    \[
        \Var_{\bx \sim \mu_{W,tv}} \langle \bx,v\rangle \ge \E_{z \sim \rho} \Var_{\bx \sim \mu_{H,z}} \langle \bx,v\rangle \mper
    \]
    Note that we can lower bound the above variance by the contribution of the isolated vertices in $H$, which are mutually independent of all other vertices.
    In particular, let $\Sisola$ be the set of isolated vertices in $H$.
    Then, we have:
    \begin{align*}
        \E_{\bz \sim \rho} \Var_{\bx \sim \mu_{H,\bz}} \langle \bx,v\rangle
        &\ge \E_{\bz \sim \rho} \Var_{\bx \sim \mu_{H,\bz}} \left[ \sum_{i \in \Sisola} \bx_i v_i \right] \\
        &= \E_{\bz \sim \rho} \sum_{i \in \Sisola} v_i^2 \cdot \Var_{\bx \sim \mu_{H,\bz}} [\bx_i].
    \end{align*}
    By assumption, the decomposition conserves the variance of arbitrary functions with parameter $\Cvar$. 
    Hence, 
    \[ \E_{\bz \sim \rho} \Var_{\bx \sim \mu_{H,\bz}}[\bx_i] \ge \Cvar \cdot \Var_{\bx \sim \mu_{W,tv}} [\bx_i] = \Cvar \cdot \left( 1 - \E_{\bx \sim \mu_{W,tv}} \left[\bx_i\right]^2 \right). \]
    Consequently, because $\|v\|_{\infty}^2 = 1/n$,
    \[ \Var_{\bx \sim \mu_{W,tv}} \langle \bx,v\rangle \ge \Cvar \cdot \left( \left( 1 - \gamma \right) - \frac{1}{n} \left\| I_{[n] \setminus H} \cdot \mean(\mu_{W, tv}) \right\|^2 \right). \]
    By assumption, $\mu_{W, tv}$ is $\entstab$-entropically stable with respect to $(x,y) \mapsto \left\|I_{[n] \setminus H} \cdot (x-y)\right\|^2$. 
    By definition, this implies that
    \begin{align*}
        \left\| I_{[n] \setminus H}  \left(\mean(\mu_{W, tv}) - \mean(\mu_{W, 0}) \right) \right\|^2 &\le \entstab \cdot \KL{\mu_{W,tv}}{\mu_{W,0}} \text{ and } \\
        \left\| I_{[n] \setminus H} \left(\mean(\mu_{W, 0}) - \mean(\mu_{W, tv}) \right) \right\|^2 &\le \entstab \cdot \KL{\mu_{W,0}}{\mu_{W,tv}}
    \end{align*}
    and since $\mean(\mu_{W, 0}) = 0$ by symmetry, we get that
    \[ \| I_{[n] \setminus H} \mean(\mu_{W, tv})  \|^2 \le \frac{\entstab}{2} \cdot \SKL{\mu_{W,tv}}{\mu_{W,0}} = \frac{\entstab}{2} \cdot \E_{\bx \sim \mu_{W,tv}} \langle tv,\bx\rangle. \]
    Therefore,
    \[
        \frac{\dif}{\dif t} \E_{\bx \sim \mu_{W,tv}} \langle \bx,v\rangle \ge \Cvar \left( \left( 1 - \gamma \right) - \frac{t\entstab}{2 n} \E_{\bx \sim \mu_{W,tv}} \langle \bx,v\rangle \right).
    \]
    Now, set
    \[
        L \coloneqq \frac{s\Cvar}{1 + \frac{s^2\Cvar\entstab}{4 n}} \cdot \left( 1 - \gamma \right).
    \]
    We claim that $\E_{\bx \sim \mu_{W, sv}} \angles{\bx, v} \ge L$.
    To prove this claim, assume for contradiction that $\E_{\bx \sim \mu_{W, sv}} \angles{\bx, v} < L$.
    Then, since $\E_{\bx \sim \mu_{W, sv}} \angles{\bx, v}$ is non-decreasing with $s$ (namely, its derivative is variance), we get:
    \begin{align*}
        \E_{\bx \sim \mu_{W,sv}} \angles{\bx, v} &\ge \Cvar \int_{0}^{s} \left( \left( 1 - \gamma \right) - \frac{tL\entstab}{2 n} \right) \dif t \\
        &= \Cvar \left( s \left( 1 - \gamma \right) - \frac{s^2 L\entstab}{4 n} \right) = L,
    \end{align*}
    which contradicts the assumption that $\E_{\bx\sim\mu_{W, sv}}\angles*{\bx, v} < L$.

    We conclude the proof by doing casework on $s$; to this end set $s_{\star}\coloneqq 2 \cdot \sqrt{\frac{n}{\Cvar\entstab}}$.
    If $s \le s_{\star}$, then the definition of $L$ yields $L \ge \frac{s\Cvar}{2} \cdot \left(1 - \gamma \right)$.
    On the other hand, the case of $s \ge s_{\star}$ immediately reduces to the above calculation because by monotonicity,
    \[
        \E_{\bx \sim \mu_{W,sv}} \angles{\bx, v} \ge \E_{\bx \sim \mu_{W,s_{\star}v}}\angles{\bx,v} \ge \left( 1 - \gamma \right) \cdot \frac{\Cvar}{2} \cdot 2 \cdot \sqrt{\frac{n}{\Cvar\entstab}}. \qedhere
    \]
\end{proof}

\begin{remark}
    The above \Cref{lem:corr-lower-bound-exp-general} can easily be generalized to relax the assumptions on $v$ to having bounded $\ell_\infty$ norm or being subgaussian. However, the resulting bound is a little messier, so we omit it for the sake of readability.
\end{remark}

As a corollary of the above and the definition of restricted Gaussian dynamics, we get the following.
\begin{corollary}[RGD boost]    \label{cor:beast-boost}
    Given $W$ as in the previous theorem, $x$ such that $\abs*{\angles*{x, v}} = r$, and $\by\sim_{P_{\RGD}}x$, we have:
    \[
        \E_{\by}\abs*{\angles*{\by, v}} \ge \left( 1 - \gamma \right) \cdot \frac{\Cvar}{2} \cdot \E_{\bg\sim\calN(0,1)} \min\left\{ \abs*{\lambda r + \sqrt{\lambda}\bg}, 2 \sqrt{\frac{n}{\Cvar\entstab}} \right\}.
    \]
\end{corollary}

\begin{lemma}[Correlation of locally stationary distributions under RGD]   \label{lem:booster-energy}
    Let $W$ be an interaction matrix satisfying the conditions in \Cref{lem:corr-lower-bound-exp-general}, and set $M = W + \lambda vv^\top$, where $v \in \left\{\pm \frac{1}{\sqrt{n}} \right\}^n$ and $\lambda \ge 1$.
    Additionally, suppose that the distribution $\mu_{W,sv}$ satisfies $\MLSI\parens*{\mu_{W,sv}} \ge n^{-1-o(1)}$ for every $s\in\R$.
    Set 
    \[
        \delta_{\star} = \delta_{\star}(\lambda) \defeq \left(1 - \gamma\right) \cdot \frac{\Cvar}{2} \cdot \min\left\{ \sqrt{\frac{2\lambda}{\pi}} , \frac{\lambda}{2} \right\}  - 1.
    \]
    Then for any $\lambda$ such that $\delta_{\star} > 0$, any $\eps$-locally stationary distribution $\nu$ under RGD, and any $\delta \in (0, \delta_{\star})$, 
    it holds that
    \[
        \E_{\bx\sim\nu} |\angles*{\bx,v}| \ge 0.99 \cdot \left( 1 - \gamma \right) \cdot \sqrt{\frac{\Cvar n}{\entstab}} - 2\delta,
    \]
    for $\eps < \delta^4 \cdot \frac{\entstab^2}{\Cvar^2(1-\gamma)^4} \cdot \frac{1}{n^2} - \exp(-\Omega(n))$.
\end{lemma}
\begin{proof}
    Set $\eta = (1-\gamma) \cdot \sqrt{\frac{\Cvar}{\entstab}}$, $\ol{\eta} = 0.99\eta$, $\eta^{\downarrow} = 0.995 \eta$, and define the function $\phi(x) \coloneqq \min\left\{\abs*{\angles*{x, v}}, \eta\sqrt{n}\right\}$.

    By \Cref{cor:bounded-function-stability}, which states that bounded functions are stable under a single step of restricted Gaussian dynamics,
    \[
        \E_{\by\sim P_{\RGD}\nu} \phi(\by) - \E_{\bx \sim \nu} \phi(\bx) < \eta\sqrt{\eps n}\mcom   \numberthis \label{eq:phi-stability}
    \]
    that is, the average value of $\phi(\bx)$ cannot change much under a single step of RGD.
    
    Assume for contradiction that $\E_{\bx\sim\nu}\phi(\bx) < \ol{\eta}\sqrt{n} - 2\delta$.
    We will show that a single step of RGD boosts the expected correlation enough that the resulting value of $\E_{\by\sim P_{\RGD}\nu}\phi(\by)$ violates \eqref{eq:phi-stability}.
    This follows from two claims.
    First, observe that by Markov's inequality:
    \[
        \Pr_{\bx \sim \nu} \left[ |\langle \bx,v\rangle| < \ol{\eta}\sqrt{n} - n^{o(1)} - \delta \right] = \Pr_{\bx\sim\nu}\left[\phi(\bx) < \ol{\eta}\sqrt{n} - n^{o(1)} - \delta\right] \ge \Omega\left(\frac{\delta}{\eta\sqrt{n}}\right).
    \]
    Second, we establish that for any $x$ such that $\abs*{\angles*{x, v}} = r$ and $\by \sim_{P_{\RGD}} x$,
    \begin{enumerate}
        \item if $r < \ol{\eta}\sqrt{n} - n^{o(1)} - \delta$, then $\E_{\by} \phi(\by)-r > \delta$, and
        \item $\E_{\by} \phi(\by) - \min\{r, \ol{\eta}\sqrt{n} - n^{o(1)}\} > -\exp(-n^{1-o(1)})$.
    \end{enumerate}

    These two claims together tell us that:
    \begin{align}
        \E_{\bx\sim P_{\RGD}\nu} \phi(\bx) - \E_{\bx\sim\nu}\phi(\bx) > \Omega\left(\frac{\delta^2}{\eta\sqrt{n}}\right) - \exp\parens*{-n^{1-o(1)}}, \numberthis \label{eq:increment}
    \end{align}
    which gives the desired contradiction, as \eqref{eq:phi-stability} and \eqref{eq:increment} cannot simultaneously be true given the assumption on the relationship between $\delta$ and $\eps$.

    It remains to prove the claimed lower bounds on $\E_{\by} \phi(\by) - r$. Let us begin with the simpler task of proving lower bounds on $\E_{\by} |\langle \by,v\rangle| - r$: we will show that
    \begin{enumerate}
        \item if $r < \eta\sqrt{n} - \delta$, then $\E_{\by} | \langle \by,v\rangle | -r > \delta$, and
        \item $\E_{\by} | \langle \by,v\rangle | - \min\{r, \eta\sqrt{n}\} > -\exp(-\Omega(n))$.
    \end{enumerate}
    Let $x$ be an arbitrary point on the hypercube with $\abs*{\angles*{x,v}} = r$.
    We may assume without loss of generality that $\angles*{x,v} > 0$.

    For ease of notation, set $\subC = \frac{\Cvar}{2} \cdot \left( 1 - \gamma  \right)$.
    We prove this lower bound by splitting into cases based on the value of $r$.
    By \Cref{cor:beast-boost}, the following lower bounds hold for $\by\sim_{P_{\RGD}}x$ depending on the value of $r$:

    \medskip
    \noindent {\bf Case $r < 1$.} We have
        \begin{align*}
            \E_{\by} \abs*{\angles*{\by, v}} &\ge \E_{\bg \sim \calN(0,1)} \min \left\{ \wt{C} \cdot |\lambda r + \sqrt{\lambda} \bg| , \eta\sqrt{n} \right\} \\
            &\ge \wt{C} \cdot \parens*{\E_{\bg\sim\calN(0,1)}\abs*{ \lambda r + \sqrt{\lambda}\bg} - \E_{\bg\sim\calN(0,1)} \abs*{ \lambda r + \sqrt{\lambda}\bg}\cdot \bone_{\wt{C} \cdot \abs*{ \lambda r + \sqrt{\lambda}\bg} > \eta\sqrt{n}}} \\
            &\ge \subC \cdot \parens*{\sqrt{\lambda}\E_{\bg\sim \calN(0,1)}|\bg| - \exp(-\Omega(n))} \\
            &= \subC \left( \sqrt{ \frac{2\lambda}{\pi} } - \exp(-\Omega(n)) \right)
        \end{align*}
        where the final inequality follows from the fact that $\E_{\bg\sim\calN(0,1)} |\bg + R|$ is minimized at $R = 0$, and from standard Gaussian concentration.

        It follows that
        \[ \E_{\by} |\langle \by,v\rangle| - r \ge \subC \cdot \sqrt{\frac{2\lambda}{\pi}} - 1 - \exp(-\Omega(n)) > \delta. \]

    \medskip
    \noindent {\bf Case $1 \le r \le \frac{3\eta\sqrt{n}}{2\subC\lambda}$.}  In this case,
        \begin{align*}
            \E_{\by}\abs*{\angles*{\by, v}} &\ge \E_{\bg\sim\calN(0,1)} \min\left\{ \frac{\subC}{2}\cdot\abs*{\lambda r + \sqrt{\lambda}\bg}, \eta\sqrt{n} \right\} \\
            &\ge \frac{\subC}{2}\cdot \parens*{\E_{\bg\sim\calN(0,1)} \abs*{\lambda r + \sqrt{\lambda}\bg} - \E_{\bg\sim\calN(0,1)} \abs*{\lambda r + \sqrt{\lambda}\bg} \cdot \bone_{\subC \cdot |\lambda r + \sqrt{\lambda} \bg| > 2\eta\sqrt{n} } } \\
            &\ge \frac{\subC\lambda}{2} \cdot r - \exp(-\Omega(n))\mper
        \end{align*}
        Observe that
        \[ \E_{\by}\abs*{\angles*{\by, v}} - r > \left(\frac{\subC\lambda}{2} - 1 \right) > \delta. \]

    \medskip
    \noindent {\bf Case $\frac{3\eta\sqrt{n}}{2\subC\lambda} \le r \le \eta\sqrt{n}-\delta$.}
    In this case, by standard Gaussian concentration, the value of $\subC \cdot \abs*{\lambda r + \sqrt{\lambda}\bg}$ exceeds $\eta\sqrt{n}$ with exponentially high probability, and hence $\E_{\by}\abs*{\angles*{\by, v}}\ge\eta\sqrt{n}-\exp(-\Omega(n))$.
    Thus, in this case too, we have $\E_{\by}\abs*{\angles*{\by, v}} - r > \delta-\exp(-\Omega(n))$.

    \medskip
    \noindent {\bf Case $\eta\sqrt{n}-\delta \le r$.}
    By identical reasoning to the previous case, in this case $\E_{\by}\abs*{\angles*{\by, v}} - \eta\sqrt{n} > -\exp(-\Omega(n))$.

    We must now translate the above lower bounds to lower bounds on $\E_{\by} \phi(\by)$. To do this, we will use \Cref{lem:mlsi-to-concentration}, which says that the correlation of a sample concentrates around its expectation. We will again consider different cases, depending on where the expectation lies. Concretely, we will prove that
    \[ \E \phi(\by) = \E \min\left\{ |\langle\by,v\rangle| , \eta\sqrt{n} \right\} \ge \min\left\{ \E |\langle\by,v\rangle| , \ol{\eta}\sqrt{n} - n^{o(1)} \right\} - \exp(-n^{1-o(1)}). \numberthis \label{eq:expec-to-phi} \]
    Towards proving this, we have
    \begin{align}
        \E \phi(\by) &= \E |\langle \by,v\rangle| \bone_{| \langle \by,v\rangle | \le \eta\sqrt{n}} + \eta\sqrt{n} \cdot \Pr\left[ | \langle \by,v\rangle | \ge \eta\sqrt{n} \right] \label{eq:expec-to-sample-1} \\
            &= \E | \langle \by,v\rangle | + \E (\eta\sqrt{n} - | \langle \by,v\rangle | ) \bone_{ | \langle \by,v\rangle | \ge \eta\sqrt{n} }. \label{eq:expec-to-sample-2}
    \end{align}
    Let $\eta^{\uparrow} = 1.005 \eta$, and recall our earlier definitions $\ol{\eta} = 0.99\eta$ and $\eta^{\downarrow} = 0.995 \eta$. 

    \medskip
    \noindent { \bf Case $\E |\langle \by,v \rangle| < \eta^{\downarrow}\sqrt{n}$.}
    
    Here, we have by \eqref{eq:expec-to-sample-2} that
    \[ \E \phi(\by) \ge \E |\langle \by,v\rangle| - \sqrt{n} \cdot \Pr\left[ |\langle \by,v\rangle| \ge \eta\sqrt{n} \right]. \]
    \Cref{lem:mlsi-to-concentration} allows us to bound
    \[ \Pr\left[ |\langle \by,v\rangle| \ge \eta\sqrt{n} \right] \le \exp\left( - \Omega \left( \frac{ \eta\sqrt{n} - \E |\langle\by,v\rangle| }{ n^{o(1)} } \right)^2 \right) = \exp\left( - n^{1-o(1)} \right). \]

    \medskip
    \noindent {\bf Case $\eta^{\downarrow} \sqrt{n} \le \E |\langle \by,v\rangle| \le \eta^{\uparrow} \sqrt{n}$. }

    \eqref{eq:expec-to-sample-2} implies that
    \begin{align*}
        \E \phi(\by) &\ge \E | \langle \by,v\rangle | - \E \left| \eta\sqrt{n} - |\langle \by,v\rangle| \right| \\
        &\ge \E | \langle \by,v\rangle | - \left| \eta\sqrt{n} - \E |\langle \by,v\rangle| \right| - \E \big| |\langle \by,v\rangle| - \E | \langle \by,v\rangle | \big| \\
        &\ge \E | \langle \by,v\rangle | - \left| \eta\sqrt{n} - \E |\langle \by,v\rangle| \right| - n^{o(1)} \\
        &\ge \eta^{\downarrow}\sqrt{n} - (\eta-\eta^{\downarrow})\sqrt{n} - n^{o(1)} \\
        &= \ol{\eta} \sqrt{n} - n^{o(1)} \mcom
    \end{align*}
    where the third inequality uses \Cref{lem:mlsi-to-concentration}.

    \medskip
    \noindent {\bf Case $\eta^{\uparrow}\sqrt{n} < \E |\langle \by,v\rangle|$. }

    \eqref{eq:expec-to-sample-1} implies that
    \begin{align*}
        \E \phi(\by) &\ge \eta\sqrt{n} \cdot \Pr\left[ | \langle \by,v\rangle | \ge \eta\sqrt{n} \right] \mper
    \end{align*}
    We give a lower bound for the above using \Cref{lem:mlsi-to-concentration} again:
    \[ \Pr\left[ |\langle \by,v\rangle| \ge \eta\sqrt{n} \right] \ge 1-\exp\left( - \Omega \left( \frac{ \E |\langle\by,v\rangle| - \eta\sqrt{n} }{ n^{o(1)} } \right)^2 \right) = 1-\exp\left( - n^{1-o(1)} \right) \mcom \]
    as desired.

    \medskip
    Now, let us put the pieces together to obtain a contradiction to \eqref{eq:phi-stability}. By \eqref{eq:expec-to-phi}, and the preceding lower bound on the shift in $\E |\langle \by, v\rangle |$, if $|\langle x,v\rangle| < \ol{\eta}\sqrt{n} - \delta - n^{o(1)}$,
    \begin{align*}
        \E_{\by \sim P_{\RGD}\delta_x} \phi(\by) &\ge \min\left\{\E |\langle \by,v\rangle| , \ol{\eta}\sqrt{n} - n^{o(1)}\right\} - \exp\left( -n^{1-o(1)} \right) \\
            &\ge \min\left\{ r+\delta , \ol{\eta}\sqrt{n} - n^{o(1)} \right\} - \exp\left( -n^{1-o(1)} \right) \\
            &\ge r+\delta - \exp\left( -n^{1-o(1)} \right) \mcom
    \end{align*}
    and similarly, if $|\angles*{x,v}| > \ol{\eta}\sqrt{n} - \delta - n^{o(1)}$, we have
    \[
        \E_{\by \sim P_{\RGD}\delta_x} \phi(\by) \ge \min\left\{ r, \ol{\eta}\sqrt{n} - n^{o(1)}\right\} - \exp\parens*{-n^{1-o(1)}}\mcom
    \]
    which completes the proof.
\end{proof}

\subsection{Glauber dynamics achieves weak recovery}\label{sec:glauber-weak-recovery}

We are finally ready to prove the main results of this section.

\mainspikedwig*
\begin{proof}
    The proof strategy is to establish that $\nu_{\bt}$ is locally stationary for restricted Gaussian dynamics, and then use \Cref{lem:booster-energy} to conclude.
    Let $\nu_t$ denote the distribution after running Glauber dynamics for time $t$.
    For $\bt\sim[0,T]$, by \Cref{thm:main-ls}, the distribution $\nu_{\bt}$ is $O\parens*{\frac{n}{\zeta T}}$-locally stationary, except with probability at most $\zeta$.
    Let us assume that this event occurs.\\
    We now compute the parameters we can plug into \Cref{lem:ls-dist-rgd}.
    \begin{itemize}
        \item First, due to local stationarity, we have $\calE\parens*{\density_{\bt},\log\density_{\bt}} < O\parens*{\frac{n}{\zeta T}}$.
        \item By \Cref{item:MLSI-bdd} of \Cref{th:bdd-mixture}, the value of $\delta$ from \Cref{lem:ls-dist-rgd} is at least $ \frac{1}{(1-\kappa)\cdot n}$.
        \item Finally, it remains to upper bound $\density_{\bt}(x)$ for all $x\in\{\pm1\}^n$. Indeed, since the Hamiltonian of $\mu$ is bounded from above by $O(n)$, we may assume that $\tau$ from \Cref{lem:ls-dist-rgd} is $O(n)$.
    \end{itemize}
    Once we plug in these parameters, we get:
    \[
        \calE_{\RGD}\parens*{ \density_{\bt}, \log\density_{\bt} } < \wt{O}\parens*{ \frac{n^{3}}{\zeta T} }\mper
    \]
    When $\zeta T \gg n^{5}$, we get
    \[
        \calE_{\RGD}\parens*{\density_{\bt}, \log\density_{\bt}} < o\parens*{\frac{1}{n^2}}\mper
    \]
    By \Cref{lem:booster-energy} with parameters $\eps = o\left(\frac{1}{n^2}\right)$ and $\delta = o(1)$, and the values for $\Cvar$ and $\entstab$ from \Cref{th:bdd-mixture}, we get that:
    \[
        \E_{\bx\sim\nu_{\bt}} \abs*{\angles*{\bx, v}} \ge \left( \sqrt{\frac{\Cvar}{\entstab}} - o(1) \right) \cdot \sqrt{n} \ge \left( \kappa\exp(-1/\kappa) - o(1)  \right) \sqrt{n}.
    \]
    We treat $\kappa$ as a constant, and can set $\zeta = o(1)$ and $T = \wt{\Omega}\parens*{n^5}$ to finish the proof.
\end{proof}

The proof of \Cref{th:sbm-recovery} for SBM weak recovery is very similar, but we provide the explicit parameter dependencies for completeness.
To be explicit, in the notation of \Cref{question:weak-recovery}, for SBM recovery we have $v = \frac{1}{\sqrt{n}} \bsigma$, $M = A_{\bG} - \frac{d}{n} \bone \bone^\top$, and $W = \ol{A}_{\bG} \coloneqq A_{\bG} - \E[A_{\bG}|\bsigma]$. 
One can verify that $M = W + \lambda \sqrt{d} \cdot vv^\top$ by using the fact that $\E[A_{\bG}|\bsigma] = \frac{d}{n} \bone \bone^\top + \frac{\lambda\sqrt{d}}{n} \bsigma \bsigma^\top$.

\mainsbm*
\begin{proof}
    Again, let $\nu_t$ denote the distribution after running Glauber dynamics on $\mu_{M, 0}$ until time $t$.
    As before, let us assume that $\nu_{\bt}$ is $O\left( \frac{n}{\zeta T} \right)$-locally stationary, which occurs with probability $1 - \zeta$.
    We now compute the parameters we can plug into \Cref{lem:ls-dist-rgd}.
    \begin{itemize}
        \item By \Cref{item:MLSI-bdd} of \Cref{th:sbm-mixture}, the value of $\delta$ from \Cref{lem:ls-dist-rgd} is at least $ \frac{1}{n^{1+o_d(1)}}$.
        \item To upper bound $\density_{\bt}(x)$ for all $x\in\{\pm1\}^n$, we have $\norm{Wx}_{\infty} \le O(\log n)$, so that we can bound $\tau \le O(n\log n)$.
    \end{itemize}
    Once we plug in these parameters, we get:
    \[
        \calE_{\RGD}\parens*{ \density_{\bt}, \log\density_{\bt} } < \wt{O}\parens*{ \frac{n^{3+o_d(1)}}{\zeta\cdot T} }\mper
    \]
    When $\zeta \cdot T \gg n^{5+o_d(1)}$, we get that
    \[ 
        \calE_{\RGD}\parens*{\density_{\bt}, \log\density_{\bt}} < o \left( \frac{1}{n^2} \right).
    \]
    By \Cref{lem:booster-energy} with parameters $\eps$ as above and $\delta = O(1)$, for some universal constant $c$,
    \[
        \E_{\bx\sim\nu_{\bt}} \abs*{\angles*{\bx, \bsigma}} \ge cn.
    \]
    Setting $\zeta = o(1)$ and $T = \wt{\Omega}\parens*{n^{5+o_d(1)}}$ completes the proof.
\end{proof}

\section*{Acknowledgments}
We would like to thank Omar Alrabiah, Sitan Chen, and Ansh Nagda for insightful discussions.

\bibliographystyle{alpha}
\bibliography{main}

\end{document}